\documentclass[english,aps,prx,superscriptaddress,floatfix, notitlepage,reprint]{revtex4-1}
\usepackage[T1]{fontenc}
\usepackage[latin9]{inputenc}
\usepackage{color}
\usepackage{array}
\usepackage{booktabs}
\usepackage{mathrsfs}
\usepackage{multirow}
\usepackage{amsmath}
\usepackage{amsthm}
\usepackage{amssymb}
\usepackage{graphicx}

\makeatletter

\providecommand{\tabularnewline}{\\}

\theoremstyle{plain}
\newtheorem{thm}{\protect\theoremname}
\theoremstyle{plain}
\newtheorem{cor}[thm]{\protect\corollaryname}

\usepackage{braket}
\usepackage{txfonts} 
\usepackage{graphicx}
\usepackage[usenames,dvipsnames]{xcolor}
\usepackage[colorlinks=true,citecolor=Blue,linkcolor=RubineRed,urlcolor=Blue]{hyperref}
\date{\today}

\makeatother

\usepackage{babel}
\providecommand{\corollaryname}{Corollary}
\providecommand{\theoremname}{Theorem}

\begin{document}
\title{Minimum-Time Quantum Control and the Quantum Brachistochrone Equation}
\author{Jing Yang \href{https://orcid.org/0000-0002-3588-0832}{\includegraphics[scale=0.05]{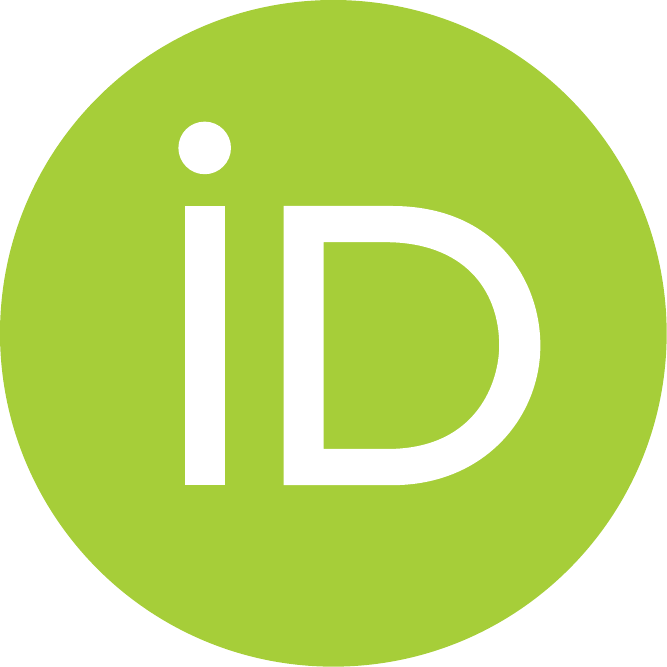}}}
\email{jing.yang@uni.lu}

\address{Department of Physics and Materials Science, University of Luxembourg,
L-1511 Luxembourg, Luxembourg}
\author{Adolfo del Campo\href{https://orcid.org/0000-0003-2219-2851}{\includegraphics[scale=0.05]{orcidid.pdf}}}
\email{adolfo.delcampo@uni.lu}

\address{Department of Physics and Materials Science, University of Luxembourg,
L-1511 Luxembourg, Luxembourg}
\address{Donostia International Physics Center, E-20018 San Sebasti\'an, Spain}
\begin{abstract}
Minimum-time quantum control protocols can be obtained from the quantum
brachistochrone formalism {[}Carlini, Hosoya, Koike, and Okudaira,
Phys. Rev. Lett. 96, 06053, (2006){]}.
We point out that the original treatment implicitly applied the variational
calculus with fixed boundary conditions. We argue that the genuine
quantum brachistochrone problem involves a variational
problem with a movable endpoint, contrary to the classical brachistochrone
problem. This formulation not only simplifies the derivation of the quantum
brachistochrone equation but introduces an additional constraint at the endpoint due to
the boundary effect. We  present the general
solution to the full quantum brachistochrone equation and discuss its main features. Using it, we prove that the speed of evolution under constraints
is reduced with respect to the unrestricted case. In addition, we
find that solving the quantum brachistochrone equation is closely connected
to solving the dynamics of the Lagrange multipliers, which is in
general governed by nonlinear differential equations. Their numerical 
integration allows  generating time-extremal
trajectories. Furthermore, when the restricted operators form a closed
subalgebra, the Lagrange multipliers become constant and the optimal
Hamiltonian takes a  concise form. The new class of analytically
solvable models for the quantum brachistochrone problem opens up the
possibility of applying it to many-body quantum systems,
exploring notions related to geometry such as quantum speed limits,  and advancing significantly
the quantum state and gate preparation for quantum information processing.
\end{abstract}
\maketitle

\section{Introduction}

The ability to control quantum systems lies at the heart of various
quantum technologies, such as quantum computation~\citep{nielsen2010quantum,nielsen2006quantum,dowling2008thegeometry},
quantum state preparation~\citep{carlini2006timeoptimal,carlini2007timeoptimal,carlini2017brachistochrone,girolami2019howdifficult,rahmani2011optimal},
quantum metrology~\citep{yuan2015optimal,pang2017optimal,yang2017quantum,yang2021variational,choi2020robustdynamic},
shortcut to adiabaticity~\citep{demirplak2003adiabatic,demirplak2005assisted,demirplak2008onthe,berry2009transitionless},
measurement-based state stabilization~\citep{mohseninia2020alwayson}, and
dynamical decoupling~\citep{viola1999dynamical,choi2020robustdynamic},
among others. 

One protocol that stands out among the different approaches to quantum control is
based on the quantum brachistochrone (QB) and was initiated by Carlini, Hosoya,
Koike, and Okudaira (CHKO)~\citep{carlini2006timeoptimal,carlini2007timeoptimal}
more than a decade ago. Motivated by the classical brachistochrone
problem, its quantum counterpart aims at finding the evolution which
takes the minimum time between two given quantum states or quantum
gates under given resources, such as a fixed norm of the Hamiltonian
and a limited  set of available Hamiltonian controls. 
The QB program by CHKO results in a differential equation with boundary conditions,
which we shall refer to as the CHKO equation in what follows. It has
inspired a number of  theoretical works~\citep{rezakhani2009quantum,koike2010timecomplexity,wang2015quantum,bender2007fasterthan,gunther2008naimarkdilated,wakamura2020ageneral,allan2021timeoptimal}
and experiments~\citep{lam2021demonstration}. The QB problem has close connections to
other fundamental notions in nonequilibrium quantum physics, such as the quantum speed limit for driven systems \cite{anandan1990geometry,uhlmann1992anenergy} and
counterdiabatic driving \cite{demirplak2003adiabatic,demirplak2005assisted,berry2009transitionless}.

\begin{figure}[t]
\begin{centering}
\includegraphics[scale=0.6]{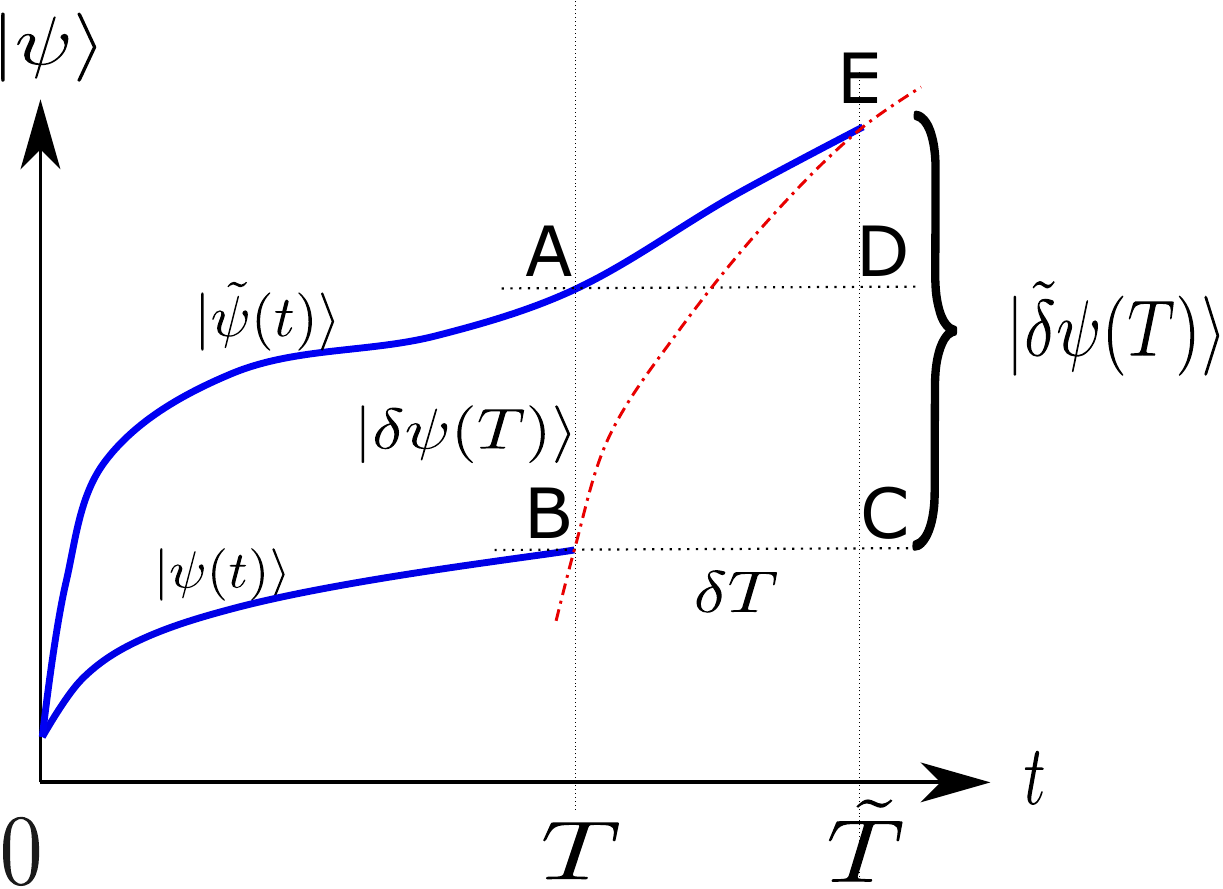}
\par\end{centering}
\caption{\label{fig:var-decomp}Decomposition of the variation of $\ket{\psi(t)}$
when the boundary at $t=T$ is moving. The geometric meaning of Eq.~(\ref{eq:varpsi})
is explained by the fact that $CE=CD+DE$.}
\end{figure}

At present, many aspects of the understanding and formulation of the QB remain to be elucidated. The QB problem was
proposed by CHKO exploiting the analogy with its classical counterpart.
The classical brachistochrone problem is solved by the variation calculus
with fixed boundary conditions \cite{gelfand2012calculus}. In this work, we reexamine the QB problem
formulated by CHKO and find the fixed boundary condition is also implicitly
assumed when performing the variational calculus. We
point out that, unlike its classical counterpart, the genuine QB problem
should be formulated as a variational problem with a movable endpoint, 
stemming from  the $U(1)$ transformation illustrated in Fig.~\ref{fig:var-decomp}.
As a result, the genuine QB problem should be solved via variational
calculus with movable boundaries. After accounting for the boundary
effect, the variational calculus with a movable boundary condition
still yields the CHKO equation as its governing equation, but with
an additional constraint at the final time,  unrecognized in preceding studies.
This constraint is necessary for clarifying the notion of locally
time-extremal trajectories. Indeed, trajectories that violate this constraint 
have been mistakenly classified as locally time-extremal in previous literature. 
Moreover, as a by-product of our formalism, the derivation of the CHKO equation for a given system
is simplified dramatically with respect to the
original approach~\citep{carlini2006timeoptimal,carlini2007timeoptimal}.

Developing efficient numerical algorithms for the
QB problem is notoriously difficult ~\citep{wang2015quantum,campaioli2019algorithm}.
Numerical calculations of the QB problem reported to date remain limited to systems
with a few qubits. Solving the QB for a many-body system, even numerically,
is a formidable task. The underlying reason for this difficulty
is the lack of an analytic understanding of the structure of the solutions
to the QB problem. Here, we analyze the general features of the full
QB equation, and find general analytical expressions for the optimal
Hamiltonian and unitary evolution operator, deriving the governing equations
of the full QB problem in terms of the dynamics of the Lagrangian
multipliers. We exactly pin down the challenge of solving the QB problem:
Determining the dynamics of the Lagrange multipliers, the dynamics of which
is governed by a set of nonlinear differential equations. 
Based on the analytical findings, we propose a recipe to generate
locally time-extremal trajectories numerically. Although our numerical
recipe cannot ascertain the optimality of the trajectories globally, it
at least makes the generation of time-extremal trajectories possible.
We believe the full QB problem may be solved numerically by combining
our algorithm here with some other searching algorithms that can select the global minimum-time trajectories among all the local
extremal ones.

It has been observed  in
the literature, though not yet rigorously proved, that increasing the number of constraints on the Hamiltonian,
as it often happens in the driving of many-body systems, slows down the preparation of a target state~\citep{bukov2019geometric,delcampo2021probing}.
Building on our analytical findings, we prove that the speed of evolution under
restricted controls in the Hamiltonian
can not exceed the speed of free evolution in the unrestricted case, in which  only the norm of
the Hamiltonian is bounded. Remarkably, we identify an important
class of analytically solvable examples of the QB problem: when the
restricted operators form a close Lie algebra, the Lagrange multipliers
become constant and both the optimal Hamiltonian and the evolution operator
take a  simple form. This opens up the possibility of solving
QB problem in many-body systems, which has manifold applications in
quantum science and technology. Our results are also of relevance to geometric approaches to quantum dynamics, with applications to quantum state and gate preparation, and quantum speed limits, for example.

We next review the basic results in Refs.~\citep{carlini2006timeoptimal,carlini2007timeoptimal}
by CHKO in Sec.~\ref{sec:The-CHKO-equation}. In Sec.~\ref{sec:Fixed-boundary-versus}, we carefully examine
that boundary condition in the variational calculus that CHKO employed
and found that as an analogy to the classical QB problem, CHKO assumed
the fixed boundary condition in the QB problem while the genuine QB
should be formulated as a variational problem with the endpoint movable.
We further show that one can derive the CHKO equation without the assumption
that the boundary conditions are fixed, significantly simplifying the calculation
in the original CHKO formalism. In Sec.~\ref{sec:The-full-QB},
we account for the effect of the movable endpoint in the CHKO action
and show that the moving boundary will introduce an additional constraint
at the final time. In Sec.~\ref{sec:governing-solution}, we derive
the optimal Hamiltonian and optimal unitary evolution for the QB problem
as well as the governing differential equations for the Lagrangian
multipliers. We prove that the speed of evolution is in general reduced
when constraints are introduced in Sec.~\ref{sec:speed} and discuss
a class of analytically solvable examples in Sec.~\ref{sec:examples}, before closing the manuscript with a summary of the main conclusions.

\section{\label{sec:The-CHKO-equation}The CHKO equation}

Throughout this work, we consider systems with a finite Hilbert space dimension $N$.
The QB problem involves finding the optimal Hamiltonian that generates a time evolution 
from the initial state $\ket{\psi_{i}}$ to
the final state $\ket{\psi_{f}}$ in the shortest possible time under given constraints. Before
formulating the problem, we first note the redundancy of the gauge
degree of freedom. If $H(t)$ is the optimal Hamiltonian, then shifting
$H(t)$ by any time-dependent scalar also yields the optimal Hamiltonian.
Fixing such a gauge degree of freedom yields the constraint 
\begin{equation}
\text{Tr}[H(t)]=0.\label{eq:traceless}
\end{equation}
Furthermore, if the norm of $H(t)$ is unbounded, one can always scale $H(t)$
so that the minimum time is zero~\citep{carlini2006timeoptimal}. This
observation motivates the norm constraint  $\text{Tr}[H^{2}(t)]\leq2\omega^{2}$.
As argued in~\citep{wang2015quantum}, when this inequality is not saturated, one can always rescale the Hamiltonian such that the trajectory
$\ket{\psi(t)}$ is unchanged and traversed in a shorter time. Therefore,
one should consider the following equality constraint,
\begin{equation}
f_{0}(H(t))=\frac{1}{2}\text{Tr}[H^{2}(t)]-\omega^{2}=0.\label{eq:norm-constraint}
\end{equation}
In addition, the controls available in a given system may be limited. This is typically the case in many-body quantum systems~\citep{sels2017minimizing,claeys2019floquetengineering,yang2021variational},
 where the  Hamiltonian $H(t)$ may lack  certain operators, such as those involving  long-range or  multiple-body interactions. This  motivates
the following constraint
\begin{equation}
f_{j}(H(t))=\text{Tr}[H(t)\mathcal{X}_{j}]=0,\,j\in[1,\,M],\label{eq:term-constraint}
\end{equation}
where $\mathcal{X}_{j}$'s are the traceless orthonormal and Hermitian
generators of the $su(N)$ Lie algebra, satisfying $\text{Tr}\left(\mathcal{X}_{i}\mathcal{X}_{j}\right)=N\delta_{ij}$.
Finally, the underlying equation of motion should be satisfied by the controlled dynamics. Under unitary evolution,  the trajectory $\ket{\psi(t)}$ is generated by the Hamiltonian $H(t)$
according to the Schr\"odinger equation 
\begin{equation}
\text{i}\ket{\dot{\psi}(t)}=H(t)\ket{\psi(t)}.\label{eq:SE}
\end{equation}
Given the above constraints, CHKO~\citep{carlini2006timeoptimal}
constructed the following action
\begin{equation}
S_{\text{CHKO}}(\ket{\psi},\,H,\,\ket{\chi},\,\lambda_{j})=\sum_{\alpha=\text{T},\,\text{C},\,\text{S}}\int_{0}^{T}L_{\text{\ensuremath{\alpha} }}dt,\label{eq:S-CHKO}
\end{equation}
involving  the time, constraint, and   Schr\"odinger Lagrangians 
defined as
\begin{align}
L_{\text{T}} & =\frac{\sqrt{g_{tt}}}{\Delta E(t)},\;L_{\text{C}}=\sum_{j}\lambda_{j}(t)f_{j}(H(t)),\label{eq:LT-def}\\
L_{\text{\ensuremath{\text{S}}}} & =\braket{\chi(t)\big|H(t)\big|\psi(t)}-\text{i}\braket{\chi(t)\big|\dot{\psi}(t)}+\text{h.c.},\label{eq:LS-def}
\end{align}
respectively. Here,  $g_{tt}\equiv\braket{\dot{\psi}(t)\big|\dot{\psi}(t)}-|\braket{\dot{\psi}(t)\big|\psi(t)}|^{2}$
is the Fubini-Study metric~\citep{braunstein1994statistical} and
$\Delta E(t)\equiv\sqrt{\text{Var}[H^{2}(t)]\big|_{\ket{\psi(t)}}}$
characterizes the speed of quantum evolution~\citep{anandan1990geometry}.
Minimization of the action $\delta S_{\text{CHKO}}=0$ yields the Euler-Lagrangian
equation
\begin{equation}
\sum_{\alpha}\frac{\partial L_{\alpha}}{\partial\bra{\psi(t)}}-\frac{d}{dt}\frac{\partial L_{\alpha}}{\partial\bra{\dot{\psi}(t)}}=0,\:\sum_{\alpha}\frac{\partial L_{\alpha}}{\partial H(t)}=0,\label{eq:EL-psi-Htilde}
\end{equation}
where $\alpha=\text{T},\,\text{C},\,\text{S}$. After performing algebraic
simplifications, CHKO arrived at the following equation
\begin{align}
\left\{ \dot{F}(t)+\text{i}[H(t),\,F(t)]\right\} \ket{\psi(t)} & =0,\label{eq:CHKO-psi}\\
\{F(t),\,\mathcal{P}(t)\} & =F(t),\label{eq:CHKO-psi-2nd}
\end{align}
where $\mathcal{P}(t)=\ket{\psi(t)}\bra{\psi(t)}$ and $F(t)=\sum_{j}\lambda_{j}(t)\partial f_{j}/\partial H(t)$.
Multiplying both sides of Eq.~(\ref{eq:CHKO-psi-2nd}) by $\mathcal{O}(t)$
from the left and taking the trace on both sides, we see that 
\begin{equation}
\langle\mathcal{O}(t)F(t)\rangle+\langle F(t)\mathcal{O}(t)\rangle=\text{Tr}[\mathcal{O}(t)F(t)],\label{eq:F-id}
\end{equation}
where $\langle\cdot\rangle$ denotes the average over the state $\ket{\psi(t)}$
throughout this work. In particular,
\begin{equation}
\langle F(t)\rangle=\text{Tr}[F(t)]=\lambda_{0}(t)\text{Tr}[H(t)]+\sum_{j\ge1}\lambda_{j}(t)\text{Tr}(\mathcal{X}_{j})=0.
\end{equation}
CHKO showed that it is sufficient to satisfy Eq.~(\ref{eq:CHKO-psi})
at all times if 
\begin{align}
\dot{F}(t)+\text{i}[H(t),\,F(t)] & =0,\label{eq:CHKO}\\
\{F(0),\,P(0)\} & =F(0),\label{eq:CHKO-initial}
\end{align}
which can be explicitly verified by noting that $F(t)=U(t)F(0)U^{\dagger}(t)$,
where $U(t)$ is the time-evolution operator generated by the Hamiltonian $H(t)$.
We refer to  Eqs.~(\ref{eq:CHKO})-(\ref{eq:CHKO-initial}) as the
CHKO equation in the folloowing. In fact, as we show in Appendix~\ref{sec:equivalence},
Eq.~(\ref{eq:CHKO-psi}) also implies Eq.~(\ref{eq:CHKO}) and therefore
they are equivalent.

\section{\label{sec:Fixed-boundary-versus}Fixed boundary versus movable boundary }

In the CHKO formalism~\citep{carlini2006timeoptimal}, the evolution
time, i.e., $\int_{0}^{T}L_{\text{T}}dt$ is optimized over all the
possible trajectories which traverse from an initial state $\ket{\psi_{i}}$
to a target state $\ket{\psi_{f}}$ under the constraints~(\ref{eq:norm-constraint})-(\ref{eq:SE}).
Note that in the constrained variational problem, before introducing
the Lagrange multipliers, the only independent function is $\ket{\psi(t)}$.

Here, we  point out that in the genuine QB problem, when
$\ket{\psi(t)}\to\ket{\tilde{\psi}(t)}=\ket{\psi(t)}+\ket{\delta\psi(t)}$,
there is also an infinitesimal change in the evolution time, which
must be accounted for in the integral upper limit of the CHKO action~(\ref{eq:S-CHKO}).
Taking into account the boundary effects requires the variational
calculus with movable boundaries, distinct from the one with fixed boundaries.
However, the Euler-Lagrangian equation~(\ref{eq:LT-def})-(\ref{eq:LS-def})
from which the CHKO equation is obtained does not contain information
on whether the boundary is moving or not. The fixed boundary
condition is implicitly assumed in the CHKO formalism. 

To see this, let us unveil the original optimization problem corresponding
to the CHKO action~(\ref{eq:S-CHKO}) before introducing the Lagrangian
multipliers. It consists of finding the extremum of 
$\int_{0}^{T}L_{\text{T}}dt$,
under the constraints that $\ket{\psi(t)}$ and $H(t)$ satisfy
Eqs.~(\ref{eq:norm-constraint})-(\ref{eq:SE}), where $T$
is kept as some constant. Note that under the constraint of the Schr\"odinger
equation~(\ref{eq:SE}), Anandan and Aharonov~\citep{anandan1990geometry}
obtained that $\sqrt{g_{tt}}=\Delta E(t)$ such that $\int_{0}^{T}L_{\text{T}}dt=T$.
It is worth noting that although  the original CHKO formalism requires 
the calculation of $\partial L_{\text{T}}/\partial\bra{\psi(t)}$,
$\partial L_{\text{T}}/\partial\bra{\dot{\psi}(t)}$ and $\partial L_{\text{T}}/\partial H(t)$ [see
e.g. Eq.~(2) of Ref.\citep{carlini2006timeoptimal}], when the endpoint $T$ is kept fixed, $\delta S_{\text{T}}$
actually vanishes regardless of the variations of $\ket{\delta\psi(t)}$
and $\ket{\delta H(t)}$,
\begin{equation}
\delta S_{\text{T}}=\delta\int_{0}^{T}L_{\text{T}}dt=\delta T=0.\label{eq:delta-ST-vanishing}
\end{equation}

Equation~(\ref{eq:delta-ST-vanishing}) not only allows us to
reproduce the CHKO equation, which validates our observation that
the endpoint is implicitly assumed in the CHKO formalism, but also
simplifies the derivation dramatically without performing the tedious
variational calculus of $\partial L_{\text{T}}/\partial\bra{\psi(t)}$,
$\partial L_{\text{T}}/\partial\bra{\dot{\psi}(t)}$ and $\partial L_{\text{T}}/\partial H(t)$.
In Eq.~(\ref{eq:EL-psi-Htilde}), one only needs to consider the
contribution from $\alpha=\text{S}$ and $\alpha=\text{C}$, ignoring
the contribution from $\alpha=\text{T}$. This results in the following
Euler-Lagrange equation 
\begin{equation}
H(t)\ket{\chi(t)}=\text{i}\ket{\dot{\chi}(t)},\label{eq:EL-sim}
\end{equation}
\begin{equation}
\,F(t)+(\ket{\psi(t)}\bra{\chi(t)}+\text{h.c.})=0.\label{eq:EL-sim-2nd}
\end{equation}
 Taking trace on both sides of Eq.~(\ref{eq:EL-sim-2nd}) yields
\begin{equation}
\braket{\psi(t)\big|\chi(t)}+\braket{\chi(t)\big|\psi(t)}=0.\label{eq:psichi-sim}
\end{equation}
Applying $\ket{\psi(t)}$ to Eq.~(\ref{eq:EL-sim-2nd}), one obtains
\begin{equation}
\ket{\chi(t)}=-\braket{\chi(t)\big|\psi(t)}\ket{\psi(t)}-F(t)\ket{\psi(t)}.\label{eq:chi-sim}
\end{equation}
Substituting Eq. (\ref{eq:chi-sim}) back to Eq.~(\ref{eq:EL-sim-2nd}),
we then find 
\begin{equation}
F(t)-\mathcal{P}(t)\braket{\chi(t)\big|\psi(t)}-\mathcal{P}(t)\braket{\psi(t)\big|\chi(t)}-\mathcal{P}(t)F(t)-F(t)\mathcal{P}(t)=0.
\end{equation}
According to Eq.~(\ref{eq:psichi-sim}), we obtain Eq.~(\ref{eq:CHKO-psi-2nd}). 

Next, we take time derivative on both sides of Eq.~(\ref{eq:EL-sim-2nd})
and using Eq.~(\ref{eq:EL-sim}), we find 
\begin{equation}
\dot{F}(t)+(\ket{\dot{\psi}(t)}\bra{\chi(t)}+\ket{\psi(t)}\bra{\dot{\chi}(t)}+\text{h.c.})=0.\label{eq:F-dot}
\end{equation}
Using the Schr\"odinger equation for $\ket{\psi(t)}$ and Eq.~(\ref{eq:EL-sim}),
yields the relation
\begin{equation}
\ket{\dot{\psi}(t)}\bra{\chi(t)}+\ket{\psi(t)}\bra{\dot{\chi}(t)}=-\text{i}[H(t),\,\ket{\psi(t)}\bra{\chi(t)}],
\end{equation}
whence it follows that
\begin{align}
 & \ket{\dot{\psi}(t)}\bra{\chi(t)}+\ket{\psi(t)}\bra{\dot{\chi}(t)}+\text{h.c.}\nonumber \\
= & -\text{i}[H(t),\,\ket{\psi(t)}\bra{\chi(t)}]+\text{i}[\ket{\chi(t)}\bra{\psi(t)},\,\tilde{H}(t)]\nonumber \\
= & -\text{i}[H(t),\,\ket{\psi(t)}\bra{\chi(t)}+\ket{\chi(t)}\bra{\psi(t)}]\nonumber \\
= & \text{i}[H(t),\,F(t)],
\end{align}
where in the last step we have used Eq.~(\ref{eq:EL-sim-2nd}). Thus,
combining the above equation with Eq.~(\ref{eq:F-dot}), we obtain
Eq.~(\ref{eq:CHKO}). The derivation of the CHKO presented here offers a dramatic simplification of the original one, 
once the proper interpretation
of the endpoints in the variational calculus is recognized. 

In Ref.~\citep{carlini2007timeoptimal}, CHKO also derived an equation
for quantum gates. The result is agrees with Eq.~(\ref{eq:CHKO}),
but without the initial condition~(\ref{eq:CHKO-initial}) and with
a different boundary condition $U(T)=\mathcal{T}\exp\left[-\text{i}\int_{0}^{T}H(\tau)d\tau\right]\sim U_{f}$,
up to some $U(1)$ phase. Following a similar procedure to the one presented here, one can also rederive
 the CHKO result for quantum gates with minimum efforts, as we show in Appendix~\ref{sec:Rederiving-CHKO-gate}.

\section{\label{sec:The-full-QB}The full quantum brachistochrone equation}

In this section, we show that in addition to the CHKO equation, the
full QB equation involves an additional constraint
at the final time $T$ stemming from the effect of the movable boundaries.
Our goal is to
\begin{equation}
\text{optimize}\,\int_{0}^{T}dt,
\end{equation}
provided that $\ket{\psi(t)}$ and $H(t)$ satisfy
the Schr\"odinger equation~(\ref{eq:SE}), and that $H(t)$ fulfills the
constraints~(\ref{eq:norm-constraint})-(\ref{eq:term-constraint})
and the boundary condition
\begin{align}
\ket{\psi(0)} & =\ket{\psi_{i}},\label{eq:initial-BC}\\
\ket{\psi(T)} & \sim\ket{\psi_{f}}.\label{eq:final-BC}
\end{align}
However, at variance with the situation in ordinary variational calculus, the boundary condition
at $T$ is not fixed. 

For a general movable boundary condition, $\ket{\psi(t)}\to\ket{\tilde{\psi}(t)}=\ket{\psi(t)}+\ket{\delta\psi(t)}$
at $T\to\tilde{T}=T+\delta T$, as shown in Fig.~\ref{fig:var-decomp},
one  readily finds that 
\begin{align}
\ket{\tilde{\delta}\psi(T)}\equiv & \ket{\tilde{\psi}(T+\delta T)}-\ket{\psi(T)}\nonumber \\
= & \ket{\tilde{\psi}(T+\delta T)}-\ket{\tilde{\psi}(T)}+\ket{\tilde{\psi}(T)}-\ket{\psi(T)}\nonumber \\
= & \ket{\dot{\tilde{\psi}}(T)}\delta T+\ket{\delta\psi(T)}\nonumber \\
= & \ket{\dot{\psi}(T)}\delta T+\ket{\delta\psi(T)}.\label{eq:varpsi}
\end{align}
The geometric meaning of Eq.~(\ref{eq:varpsi}) is ilustrated in Fig.~\ref{fig:var-decomp}.
In general, the boundary condition dictates how $\ket{\tilde{\psi}(\tilde{T})}$
changes and therefore determines $\ket{\tilde{\delta}\psi(T)}$. In
the current context, the final boundary condition~(\ref{eq:final-BC})
dictates that 
\[
\ket{\tilde{\psi}(\tilde{T})}=e^{\text{i}\delta\theta(T)}\ket{\psi(T)},
\]
where $\delta\theta(T)$ should be an arbitrary variation. We emphasize
that the value of the variational trajectory $\ket{\tilde{\psi}(t)}$
at the new final time is proportional to the old trajectory $\ket{\psi(t)}$
at the old final time with the proportionality constant being a phase
close to the identity. Therefore, we conclude that
\begin{eqnarray}
\ket{\delta\psi(T)} & = & -\ket{\dot{\psi}(T)}\delta T+\text{i}\delta\theta(T)\ket{\psi(T)}.\label{eq:delta-psiT}
\end{eqnarray}
Finally, we remark that that $\ket{\delta\psi(t)}$  introduces
a change of the total time $\delta T$ at the end point. Without introducing
the Lagrange multipliers, $\ket{\delta H(t)}$ depends on $\ket{\delta\psi(t)}$.
In this case,  a change in the Hamiltonian  leads to a change of the
evolution time $T$. However, after introducing the  Lagrange multipliers,
$H(t)$ and $\ket{\psi(t)}$ are independent functions. Only the variation
of $\ket{\psi(t)}$  leads to a change of the total evolution time. 

Given these considerations, the variation of the CHKO action 
reads \begin{widetext}
\begin{eqnarray}
\delta S_{\text{CHKO}} & =&\delta T+\int_{0}^{T+\delta T}L_{\text{S}}\left(\ket{\psi(t)}+\ket{\delta\psi(t)}\right)dt-\int_{0}^{T}L_{\text{S}}\left(\ket{\psi(t)}+\ket{\delta\psi(t)}\right)dt\nonumber \\
 & & +\int_{0}^{T}L_{\text{S}}\left(\ket{\psi(t)}+\ket{\delta\psi(t)}\right)dt-\int_{0}^{T}L_{\text{S}}\left(\ket{\psi(t)}\right)dt\nonumber \\
 & =&\delta T+L_{\text{S}}\left(\ket{\psi(T)}\right)\delta T+\int_{0}^{T}\langle\delta\psi(t)\bigg|\left(\frac{\partial L_{\text{S}}}{\partial\bra{\psi(t)}}-\frac{d}{dt}\frac{\partial L_{\text{S}}}{\partial\bra{\dot{\psi}(t)}}\right)\bigg\rangle dt+\langle\delta\psi(t)\bigg|\frac{\partial L_{\text{S}}}{\partial\bra{\dot{\psi}(t)}}\bigg\rangle\bigg|_{t=0}^{t=T}.
\end{eqnarray}
\end{widetext}
Substituting Eq.~(\ref{eq:delta-psiT}) and noting
that 
\begin{align}
\ket{\delta\psi(0)} & =0,\\
L_{\text{S}}\left(\ket{\psi(T)},\,\ket{\dot{\psi}(T)}\right) & =0,\\
\frac{\partial L_{\text{S}}}{\partial\bra{\psi(t)}}-\frac{d}{dt}\frac{\partial L_{\text{S}}}{\partial\bra{\dot{\psi}(t)}} & =H(t)\ket{\chi(t)}-\text{i}\ket{\dot{\chi}(t)},
\end{align}
we arrive at 
\begin{align}
\delta S_{\text{CHKO}} & =\int_{0}^{T}\langle\delta\psi(t)\bigg|\left(\frac{\partial L_{\text{S}}}{\partial\bra{\psi(t)}}-\frac{d}{dt}\frac{\partial L_{\text{S}}}{\partial\bra{\dot{\psi}(t)}}\right)\bigg\rangle dt\nonumber \\
 & +\left(1-\text{i}\braket{\dot{\psi}(T)\big|\chi(T)}\right)\delta T+\text{i}\delta\theta(T)\braket{\psi(T)\big|\chi(T)}.
\end{align}
Thus $\delta S_{\text{CHKO}}=0$ yields not only Eq.~(\ref{eq:EL-sim}),
but also two additional equations, 
\begin{align}
\text{i}\braket{\dot{\psi}(T)\big|\chi(T)} & =1,\label{eq:state-add1}\\
\braket{\psi(T)\big|\chi(T)} & =0.\label{eq:state-add2}
\end{align}
Using the Schr\"odinger equation into Eq.~(\ref{eq:state-add1})
yields 
\begin{equation}
\braket{\psi(T)\big|H(T)\big|\chi(T)}+1=0,
\end{equation}
which together with Eq.~(\ref{eq:chi-sim}) and
Eq.~(\ref{eq:state-add2}) leads to
\begin{equation}
\braket{\psi_{f}\big|H(T)F(T)\big|\psi_{f}}=1.\label{eq:BC}
\end{equation}
Equation~(\ref{eq:BC}) is one of our central results and constitutes an additional
non-trivial constraint due to the moving boundary, unrecognized
in the previous literature. An analogous constraint also exists in the context of generating a target quantum gate, see Appendix~\ref{sec:FullQB-gate}. 

The reader may wonder what happens when both the initial and  the final boundary
conditions are movable, i.e., $\ket{\psi_{i}(0)}=e^{\text{i}\theta_{i}}\ket{\psi_{i}}$
and $\ket{\psi(T)}=e^{\text{i}\theta_{f}}\ket{\psi_{f}}$. Obviously,
trajectories satisfying these boundary conditions can be identified
with trajectories satisfying $\ket{\psi_{i}(0)}=\ket{\psi_{i}}$ and
$\ket{\psi(T)}=e^{\text{i}(\theta_{f}-\theta_{i})}\ket{\psi_{f}}$,
with the Hamiltonian and time-evolution operator remaining unchanged. This is  tantamount to imposing 
the boundary conditions in Eqs. (\ref{eq:initial-BC})-(\ref{eq:final-BC}).
Therefore, it is sufficient to consider the case where the initial
state is fixed while the final boundary condition is movable, according
to the $U(1)$ gauge transformation.

\section{\label{sec:governing-solution}The governing equations for the full
Quantum brachistochrone equation}

Having found the additional constraint due to the effect of a moving
boundary, we now discuss how to solve the complete set of QB equations,
including the CHKO equation~(\ref{eq:CHKO})-(\ref{eq:CHKO-initial}),
the constraints~(\ref{eq:traceless})-(\ref{eq:term-constraint}),
the boundary conditions in Eq.~(\ref{eq:final-BC}), and Eq.~(\ref{eq:BC}).
We shall divide the solution process into two stages: 

(i) In the first stage, we consider $\lambda_{j}(t)$$(j\ge0)$ and
$H(t)$ as unknown functions, treating $\lambda_{j}(0)$ and $H(0)$
as  fixed  and solve the constraints~(\ref{eq:traceless})-(\ref{eq:term-constraint})
together with Eq.~(\ref{eq:CHKO}). We derive the expression  of $H(t)$ and $U(t)$
in terms of the Lagrange multipliers $\lambda_{j}(t)(j\ge0)$, whose
dynamics is given by a nonlinear differential equation.

(ii) In the next stage, we consider $\lambda_{j}(0)$ $(j\ge0)$, $H(0)$
and $T$ as unknowns and solve Eq.~(\ref{eq:CHKO-initial}) subject to
the boundary conditions in Eq.~(\ref{eq:final-BC}) and Eq.~(\ref{eq:BC}).

To gain some qualitative understanding of the solution, let us count
the number of constraints and the number of unknowns in the CHKO equation
at both stages. In the first stage, the number of independent equations
is $N^{2}+M+1$, which is the same as the number of unknowns. Since
the norm constraint~(\ref{eq:norm-constraint}) is nonlinear, the
number of solutions, provided that they exist, should be multiple in general. In
fact, as one can see from Eqs.~(\ref{eq:lam-0})-(\ref{eq:eta-exp}),
the differential equations involving $\lambda_{j}(t)$ are highly
nonlinear. In the second stage, the number of independent equations
is $N^{2}+M+4$, while the number of unknowns is $N^{2}+M+2$. Therefore,
we see that the presence of Eq.~(\ref{eq:BC}) imposes a compatibility
condition among the coefficients $\lambda_{j}(0)$ and $H(0)$. Violation of 
the compatibility condition could mean that the extremal-time trajectory does not exist, which can be expected  when the Hamiltonian is highly
restricted. For example, if the Hamiltonian is local, it may not  be able to generate a trajectory between
an initially separable state and a final entangled state. 

\begin{table*}
\begin{centering}
\begin{tabular}{ccccccc}
\toprule 
 & \multicolumn{2}{c}{Stage (i)} & \multicolumn{4}{c}{Second (ii)}\tabularnewline
\midrule
\# of unknowns & \multicolumn{2}{c}{$\lambda_{j}(t)(j\ge1)$: \#=$M+1$, $H(t)$:\#$=N^{2}$} & \multicolumn{4}{c}{$\lambda_{j}(0)(j\ge0)$ \#$=M+1$, $H(0)$:\#$=N^{2}$, $T$: \#$=1$}\tabularnewline
\midrule 
\multirow{2}{*}{\# of independent equations} & Eqs.~(\ref{eq:traceless}-\ref{eq:term-constraint}) & Eq.~(\ref{eq:CHKO-initial}) & Eqs.~(\ref{eq:traceless}-\ref{eq:term-constraint}) at $t=0$ & Eq.~(\ref{eq:CHKO-initial}) & Eq.~(\ref{eq:final-BC}) & Eq.~(\ref{eq:BC}) \tabularnewline
\cmidrule{2-7} \cmidrule{3-7} \cmidrule{4-7} \cmidrule{5-7} \cmidrule{6-7} \cmidrule{7-7} 
 & $M+2$ & $N^{2}-1$ & $M+2$ & $(N-1)^{2}$ & $2N-1$ & $2$\tabularnewline
\bottomrule
\end{tabular}
\par\end{centering}
\caption{\label{tab:EqAnalysis}Analysis of the number of equations and the
number of unknowns in solving the full QBE. Note that $\text{Tr}[H(t)]=0$ 
implies $\text{Tr}[F(t)]=0$.  The number of independent equations
in Eqs.~(\ref{eq:CHKO}) is  then $N^{2}-1$ rather than $N^{2}$. Similarly,
Eq.~(\ref{eq:CHKO-initial}) is equivalent to $\braket{e_{1}\big|F(0)\big|e_{1}}=0$
and $\braket{e_{k}\big|F(0)\big|e_{l}}=0$, where $k\ge2$. However, the
former condition is guaranteed as $\text{Tr}[H(0)]=0$, given that $\braket{e_{1}\big|F(0)\big|e_{1}}=\text{Tr}[F(0)]$.
Therefore, the number of independent equations in Eq.~(\ref{eq:CHKO-initial})
is $(N-1)^{2}$. Taking into account the redundant overall phase in
the state, and the normalization of $\ket{\psi(T)}$ and $\ket{\psi_{f}}$,
Eq.~(\ref{eq:final-BC}) imposes $2N-1$ constraints. }

\end{table*}

Takahashi~\citep{takahashi2013howfast} found the solution to Eq.~(\ref{eq:CHKO}),
using Lewis-Riesenfeld invariants~\citep{gungordu2012dynamical,lewis1969anexact}.
Using the fact that
the solution in the first stage admits a compact
form, it is  shown in Appendix~\ref{sec:Derivation-Governing-Eqs} that the Hamiltonian and the evolution operator are respectively
given by
\begin{align}
H(t) & =\frac{1}{\lambda_{0}(t)}\left[\lambda_{0}(0)\tilde{H}(t)+\sum_{j\ge1}\lambda_{j}(0)\tilde{\mathcal{X}}_{j}(t)\right]-G(t),\label{eq:H-expression}\\
U(t) & =V(t)\exp\left(-\text{i}\left[\lambda_{0}(0)H(0)+\sum_{j\ge1}\lambda_{j}(0)\mathcal{X}_{j}\right]\int_{0}^{t}\frac{d\tau}{\lambda_{0}(\tau)}\right).\label{eq:U-expression}
\end{align}
Here, $\tilde{\mathcal{O}}(t)=V(t)\mathcal{O}(0)V^{\dagger}(t)$ is
the operator in the frame generated by the restricted operators, $V(t)$
satisfies the Schr\"odinger-like equation
\begin{equation}
\dot{V}(t)=\text{i}G(t)V(t),\label{eq:V-SE}
\end{equation}
with the initial condition $V(0)=\mathbb{I}$ and the generator
\begin{equation}
G(t)\equiv\frac{\sum_{j\ge1}\lambda_{j}(t)}{\lambda_{0}(t)}\mathcal{X}_{j}. \label{eq:G-def}
\end{equation}
Further,  $\lambda_{j}(t)$ satisfy
\begin{equation}
\lambda_{0}(t)\dot{\lambda}_{0}(t)=-\frac{1}{2\omega^{2}}\sum_{j\ge1}\sum_{l\ge1}\lambda_{j}(t)\lambda_{l}(t)\eta_{jl}(t),\label{eq:lam-0}
\end{equation}
\begin{equation}
\dot{\lambda}_{j}(t)=\frac{1}{N}\sum_{l\ge1}\lambda_{l}(t)\eta_{jl}(t),\label{eq:lam-j}
\end{equation}
where $\mathscr{X}_{jl}\equiv\text{i}[\mathcal{X}_{j},\,\mathcal{X}_{l}]$
and
\begin{equation}
\eta_{jl}(t)=\text{Tr}[H(t)\mathscr{X}_{jl}].\label{eq:eta-exp}
\end{equation}
Equations~(\ref{eq:H-expression})-(\ref{eq:lam-j}) constitute another central
results of this work. We note an important symmetry of Eqs.~(\ref{eq:H-expression})-(\ref{eq:lam-j}).
\begin{align}
\lambda_{0}(t)\to\tilde{\lambda}_{j}(t)=\frac{1}{c}\lambda_{j}(t),\, & \eta_{jl}(t)\to\tilde{\eta}_{jl}(t)=\eta_{jl}(t),\label{eq:sym-a}\\
H(t)\to\tilde{H}(t)=H(t),\, & U(t)\to\tilde{U}(t)=U(t),\label{eq:sym-b}
\end{align}
 as long as $c\neq0$. 

In the second stage, one may assume that $\lambda_{0}(0)=0$. Therefore
\begin{align}
F(0) & =\sum_{j\in Y}\mu_{j}\mathcal{Y}_{j}+\sum_{j\in X}\lambda_{j}(0)\mathcal{X}_{j},\nonumber \\
H(0) & =\sum_{j\in Y}\mu_{j}\mathcal{Y}_{j},\label{eq:H0-choice}
\end{align}
where $\mathcal{Y}_{j}$ are the  allowed orthonormalgenerators, and $X$
and $Y$ denote the sets of indices for the disallowed operators and the
allowed operators, respectively. Furthermore, the norm constraint implies
that 
\begin{equation}
N\sum_{j\in Y}\mu_{j}^{2}=2\omega^{2}.\label{eq:mu-j-norm}
\end{equation}
We define 
\begin{equation}
\ket{\tilde{\psi}_{f}^{\perp}}\equiv\ket{\psi_{f}}-\braket{\psi_{i}\big|\psi_{f}}\ket{\psi_{i}}\label{eq:psif-perp-til}
\end{equation}
and 
\begin{equation}
\ket{\psi_{f}^{\perp}}\equiv\frac{\ket{\tilde{\psi}_{f}}}{\|\ket{\tilde{\psi}_{f}}\|}=\frac{\ket{\tilde{\psi}_{f}}}{\sin\Omega_{\text{B}}},\label{eq:psif-perp}
\end{equation}
where 
\begin{align}
\cos\Omega_{\text{B}} & \equiv|\braket{\psi_{i}\big|\psi_{f}}|\in[0,\,1],\\
\phi & \equiv\text{arg}(\braket{\psi_{i}\big|\psi_{f}}),\label{eq:phi-def}
\end{align}
and $\Omega_{\text{B}}\in[0,\,\pi/2]$ is the Bures angle~\citep{nielsen2010quantum}.
The remaining orthonormal basis is denoted by $\{\ket{e_{k}}\}_{k=1}^{N}$
with $\ket{e_{1}}=\ket{\psi_{i}}$ and $\ket{e_{2}}=\ket{\psi_{f}^{\perp}}$.
Equation~(\ref{eq:CHKO-initial}) indicates that $F(0)$ should have the following representation in the orthonormal basis $\{\ket{e_{k}}\}_{k=1}^{N}$,
\begin{equation}
F(0)=\left[\begin{array}{c|clcc}
0 & \times & \times & \cdots & \times\\
\hline \times\\
\times\\
\vdots &  &  & \text{\huge0}\\
\times
\end{array}\right],\label{eq:F-structure}
\end{equation}
where $\times$ denotes  matrix elements that are in general not
zero. Thus, the values of $\mu_{j}(0)$ and $\lambda_{j}(0)$ are chosen
such that Eqs.~(\ref{eq:F-structure}) and (\ref{eq:final-BC})
are satisfied. Equation~(\ref{eq:BC}) can be split into two parts
\begin{align}
\text{Re}\braket{\psi_{f}\big|H(T)F(T)\big|\psi_{f}} & =1,\label{eq:Re-BC}\\
\text{Im}\braket{\psi_{f}\big|H(T)F(T)\big|\psi_{f}} & =0.\label{eq:IM-BC}
\end{align}
The first part~(\ref{eq:Re-BC}) can be always satisfied by setting $c=\text{Re}\braket{\psi_{f}\big|H(T)F(T)\big|\psi_{f}}$
in the symmetry transformation~(\ref{eq:sym-a})-(\ref{eq:sym-b})
and renormalizing $\lambda_{j}(t)$'s. However, the second part, Eq.~(\ref{eq:IM-BC})
cannot be gauged away by renormalizing  $\lambda_{j}(t)$'s and
therefore imposes a nontrivial constraint on the Lagrange multipliers.
In fact, whether Eq.~(\ref{eq:IM-BC}) holds or not does not depend on the
choice of $\lambda_{j}(t)$ or $\tilde{\lambda}_{j}(t)$. This constraint
has been ignored previously in the literature. 

With Eq.~(\ref{eq:CHKO}), Eq.~(\ref{eq:IM-BC}) can be rewritten
as 
\begin{equation}
\braket{\psi_{f}\big|[H(T),\,G(T)]\big|\psi_{f}}=0,
\end{equation}
where we have used $[H(t),\,F(t)]=\lambda_{0}(t)[H(t),\,G(t)]$. In
terms of the initial state 
\begin{equation}
\braket{\psi_{i}\big|[\text{i}U^{\dagger}(T)\dot{U}(T),\,U^{\dagger}(T)G(T)U(T)]\big|\psi_{i}}=0.\label{eq:BConstraint-initial}
\end{equation}

Previous works have shown that solving the CHKO equation numerically
is notoriously difficult~\citep{wang2015quantum,campaioli2019algorithm}.
The problem is particularly complex at the many-body level.
Yet, we note that the solution process in the second stage is
almost trivial since all the equations are algebraic equations about
$\lambda_{j}(0)$ and $H(0)$.  The challenges in numerically solving
QB problem arise from determining self-consistently the dynamics of
the Lagrange multipliers, which are governed by the nonlinear differential
equation Eqs.~(\ref{eq:lam-0})-(\ref{eq:lam-j}). This is the reason
why analytic examples of the QB problem are very rare and remain limited
to very simple cases. Nevertheless, using the results above we report
a class of new analytic examples of the QB in Sec.~\ref{sec:examples}.

Next, instead of solving the QB problem completely, we propose a method
to generate time-extremal trajectories numerically. 
To reach this goal, we first leave the final state undetermined; it will be eventually
specified by imposing Eq.~(\ref{eq:final-BC}).
In leaving the final state unfixed, one can focus on the highly
nontrivial part of solving the QB problem, i.e., determining the dynamics
of the Lagrange multipliers:
\begin{enumerate}
\item For an initial state $\ket{\psi_{i}}$, we choose an initial Hamiltonian 
$H(0)$ that bears the form of Eqs.~(\ref{eq:H0-choice})-(\ref{eq:mu-j-norm})
so that it satisfies Eqs.~(\ref{eq:traceless})-(\ref{eq:term-constraint})
at $t=0$. Eq.~(\ref{eq:CHKO-initial}) implies that $F(0)$ must
have the structure of Eq.~(\ref{eq:F-structure}), which  introduces
additional constraints between $\mu_{j}$'s and $\lambda_{j}(0)$'s,
as discussed in Appendix~\ref{sec:Initial-form}. Section~\ref{subsec:multiple-constraints}
provides an example of how this step is performed in an analytic example
where the dynamics of the Lagrange multipliers are constants.
\item Choosing the initial values of $\mu_{j}$'s and $\lambda_{j}(0)$
that satisfy the constraints in Step $1$, one can generate the time-optimal
trajectories by numerically integrating Eqs.~(\ref{eq:V-SE}-\ref{eq:lam-j}).
The numerical integration will stop until it reaches some time $T$
such that Eq.~(\ref{eq:IM-BC}) is satisfied with $\text{Re}\braket{\psi(T)\big|H(T)F(T)\big|\psi(T)}\neq0$.
We note that Eq.~(\ref{eq:Re-BC}) can be satisfied by choosing $c=\text{Re}\braket{\psi(T)\big|H(T)F(T)\big|\psi(T)}$
and then renormalizing $\lambda_{j}(t)$ to $\tilde{\lambda}_{j}(t)$.
Upon setting $\ket{\psi_{f}}$ equal to $\ket{\psi(T)}$, we find
a time-extremal trajectory between $\ket{\psi_{i}}$ and $\ket{\psi_{f}}$.
\end{enumerate}
Although our numerical recipe here does not give the optimality of
the trajectories globally, it makes the generation of time-extremal
trajectories possible. The full QB problem may be solved
numerically by combining our algorithms here with some other searching
algorithms that can select the global minimum-time trajectories
among all the local extremal ones.

\section{Free evolution}

To illustrate how the QB solution can be found by making use of the two stages presented
in Sec.~\ref{sec:governing-solution}, we consider the simplest case
with $M=0$. We refer to this case as the free evolution, since it is
free from the operator constraint~(\ref{eq:norm-constraint}), and refer to the
case $M\ge1$ as the operator-restricted evolution, given that it is subject to the operator constrain~(\ref{eq:term-constraint}). 
The free evolution was previously discussed by CHKO~\citep{carlini2006timeoptimal}.
However, several subtleties in the problem are not discussed by CHKO,
including the constraint of the moving boundary effect~(\ref{eq:BC}). 

Since in this case $\tilde{\lambda}_{j}(t)=0,\,j\ge1$, the dynamics
of the Lagrangian multiplier in the first stage becomes trivial, $\tilde{\lambda}_{0}(t)=\tilde{\lambda}_{0}(0)\equiv\tilde{\lambda}_{0}$.
Furthermore, $V(t)=\mathbb{I}$ and therefore Eqs.~(\ref{eq:H-expression})-(\ref{eq:U-expression})
become 
\begin{align}
H_{\text{F}}(t) & =H_{\text{F}},\\
U_{\text{F}}(t) & =e^{-\text{i}H_{\text{F}}t},
\end{align}
where $H_{\text{F}}$ is some time-independent Hamiltonian. The
solution in the first stage readily follows. Let us now
discuss the solution in the second stage. Equation~(\ref{eq:CHKO-initial})
becomes 
\begin{equation}
H_{\text{F}}\ket{\psi_{i}}\bra{\psi_{i}}+\ket{\psi_{i}}\bra{\psi_{i}}H_{\text{F}}=H_{\text{F}}.\label{eq:matrix-eq}
\end{equation}
For $N=2$, according to Eq.~(\ref{eq:F-structure}), in the orthonormal
basis $\{\ket{\psi_{i}},\,\ket{\psi_{f}^{\perp}}\}$, $H_{\text{F}}$
becomes
\begin{equation}
H_{\text{F}}=\left[\begin{array}{c|c}
0 & h_{if}\\
\hline h_{if}^{*} & 0
\end{array}\right].\label{eq:H0}
\end{equation}
Eq.~(\ref{eq:norm-constraint}) implies that $|h_{if}|^{2}=\omega^{2}$,
so we can denote $h_{if}=\omega e^{-\text{i}\varphi}$. Therefore
\begin{equation}
H_{\text{F}}=\omega(\cos\varphi\sigma_{\text{eff}}^{x}+\sin\varphi\sigma_{\text{eff}}^{y}),\label{eq:d2-opt-Htilde}
\end{equation}
where 
\begin{align}
\sigma_{\text{eff}}^{x} & \equiv\ket{\psi_{i}}\bra{\psi_{f}^{\perp}}+\ket{\psi_{f}^{\perp}}\bra{\psi_{i}},\\
\sigma_{\text{eff}}^{y} & \equiv-\text{i}\left(\ket{\psi_{i}}\bra{\psi_{f}^{\perp}}-\ket{\psi_{f}^{\perp}}\bra{\psi_{i}}\right).
\end{align}
In the basis  $\{\ket{\psi_{i}},\,\ket{\psi_{f}^{\perp}}\}$, 
\begin{align}
\ket{\psi_{i}}=\begin{pmatrix}1\\
0
\end{pmatrix},\; & \ket{\psi_{f}^{\perp}}=\begin{pmatrix}0\\
1
\end{pmatrix},\,\ket{\psi_{f}}=\begin{pmatrix}\cos\Omega_{\text{B}}e^{\text{i}\phi}\\
\sin\Omega_{\text{B}}
\end{pmatrix},\label{eq:psi-ref}
\end{align}
and Eq.~(\ref{eq:d2-opt-Htilde}) becomes $H_{\text{F}}=\omega(\cos\varphi\sigma_{\text{eff}}^{x}+\sin\varphi\sigma_{\text{eff}}^{y})$.
It is then straightforward to compute 
\begin{equation}
e^{-\text{i}H_{\text{F}}T}\ket{\psi_{i}}=\begin{pmatrix}\cos(\omega T)\\
-\text{i}\sin(\omega T)e^{\text{i}\varphi}
\end{pmatrix}.
\end{equation}
Thus, the boundary condition~(\ref{eq:final-BC}) can be satisfied
if and only if 
\begin{align}
T & =\frac{\Omega_{\text{B}}}{|\omega|}+\frac{2\pi k}{|\omega|},\,k\in\mathbb{N},\\
\varphi & =2\pi l-\left(\phi-\frac{\pi}{2}\right),\,l\in\mathbb{Z}.
\end{align}
The optimal Hamiltonian is thus given by 
\begin{equation}
H_{\text{F}}=\omega(\sin\phi\sigma_{\text{eff}}^{x}-\cos\phi\sigma_{\text{eff}}^{y}),\label{eq:H-free}
\end{equation}
with the global minimum time being $\Omega_{\text{B}}/\omega$. Upon defining
$\ket{\psi_{f}^{\prime\perp}}=\text{i}e^{-\text{i}\phi}\ket{\psi_{f}^{\perp}}$,
Eq.~(\ref{eq:H-free}) reduces to Eq.~(13) in Ref.~\citep{carlini2006timeoptimal}.
In this work, we shall stick to Eq.~(\ref{eq:H-free}) because it
explicitly displays the role of the phase angle $\phi$, which plays
a role in the restricted evolution, as we shall see in Sec.~\ref{sec:examples}.
It is also worth noting that when $\braket{\psi_{i}\big|\psi_{f}}=0$,
$\phi$ can be chosen arbitrarily and there is an infinite family of
optimal Hamiltonians in this case. 

Furthermore, we note that
\begin{equation}
\text{Re}\braket{\psi_{f}\big|H_{\text{F}}\big|\psi_{f}}=\frac{1}{\tilde{\lambda_{0}}}\braket{\psi_{f}\big|H_{\text{F}}^{2}\big|\psi_{f}}=2\tilde{\lambda}_{0}\omega^{2}=1.
\end{equation}
Equation~(\ref{eq:Re-BC}) indicates that $\tilde{\lambda}_{0}=1/(2\omega^{2})$.
In addition, for this particular example, Eq.~(\ref{eq:IM-BC})
is satisfied automatically. The fact that Eq.~(\ref{eq:Re-BC})
imposes the exact value of $\lambda_{0}$ was noted
in~\citep{carlini2006timeoptimal}. 

Finally, we argue that Eq.~(\ref{eq:H-free}) is also the optimal
Hamiltonian for general $N$-level systems. The argument builds on the fact that
\begin{equation}
T=\int\frac{ds}{\Delta E(t)}=\frac{\int ds}{\omega},\label{eq:T-free}
\end{equation}
and thus the minimum-time trajectory is also the minimum-length trajectory.
On the other hand, any trajectory that is outside of the subspace
$\text{span}\{\ket{\psi_{i}},\,\ket{\psi_{f}^{\perp}}\}$ will take longer
than its projected trajectory onto this subspace.
There may not be a unique way of constructing the target Hamiltonian
that generates the projected trajectory. For example, with the method
of counter-diabatic driving \citep{demirplak2003adiabatic,demirplak2005assisted,demirplak2008onthe,berry2009transitionless},
one can construct the generating Hamiltonian for any given trajectory
$\ket{\psi(t)}$ as follows: First, one  constructs a set of orthonormal
trajectories $\{\ket{\psi_{n}(t)}\}$, $\langle \psi_{n}(t)|\psi_{m}(t)\rangle=\delta_{nm}$
with $\ket{\psi_{0}(t)}=\ket{\psi(t)}$. Then, the target Hamiltonian
is derived as $H(t)=\text{i}\dot{U}(t)U^{\dagger}(t)$ with $U(t)=\sum_{n}\ket{\psi_{n}(t)}\bra{\psi_{n}(t)}$.

Thus, for general $N$-level systems, it suffices to consider the subspace
$\text{span}\{\ket{\psi_{i}},\,\ket{\psi_{f}^{\perp}}\}$, with the
optimal Hamiltonian in this subspace being also given by Eq.~(\ref{eq:H-free}).

\section{\label{sec:speed}The speed of evolution under constraints}

Having discussed the general solutions to the QB problem, let us calculate
the speed of evolution according to Eqs.~(\ref{eq:H-expression},~\ref{eq:U-expression}).
The importance of the speed of evolution cannot be overemphasized
in quantum information processing. For example, it is generally conjectured
that with more constraints, the speed of evolution will be reduced
in general when compared to the free evolution~\citep{carlini2006timeoptimal,bukov2019geometric}.
Nevertheless, a rigorous and systematic study on how the speed of
evolution for time-optimal trajectories is affected under constraints
has not been reported in the literature, to the best of our knowledge. Next, we rigorously prove this
assertion.
\begin{thm}
\label{thm:speed}The speed of evolution under constraints, in general,
can not exceed $\omega$.
\end{thm}

\begin{proof}
The speed of evolution  can be rewritten as 
\begin{equation}
\Delta E^{2}(t)=\text{Var}\left[F(t)/\lambda_{0}(t)-G(t)\right]\big|_{\ket{\psi(t)}}\,.
\end{equation}
Using Eq.~(\ref{eq:F-id}), one finds that 
\begin{align}
 & \langle F^{2}(t)\rangle-\lambda_{0}(t)(\langle F(t)G(t)\rangle+\langle G(t)F(t)\rangle)\nonumber \\
= & \frac{1}{2}\text{Tr}[F^{2}(t)]-\lambda_{0}(t)\text{Tr}[F(t)G(t)]\nonumber \\
= & \frac{1}{2}\text{Tr}[F^{2}(t)]-\lambda_{0}^{2}(t)\text{Tr}[G^{2}(t)]\nonumber \\
= & \omega^{2}\lambda_{0}^{2}(t)-\frac{1}{2}\lambda_{0}^{2}(t)\text{Tr}[G^{2}(t)].
\end{align}
Therefore, 
\begin{equation}
\Delta E^{2}(t)=\omega^{2}-\left(\text{Tr}[G^{2}(t)]/2-\text{Var}[G(t)]\big|_{\ket{\psi(t)}}\right).
\end{equation}
We recall the following inequality, often used in quantum metrology~\citep{giovannetti2006quantum,Boixo07,Chenu17,Beau17,yang2021variational},
\begin{equation}
\text{Var}[G(t)]\big|_{\ket{\psi(t)}}\leq\left[\frac{g_{\max}(t)-g_{\min}(t)}{2}\right]^{2}
\end{equation}
where $g_{k}(t)$ is the eigenvalue of $G(t)$. 
Thanks to it, 
\begin{align}
 & \text{Tr}[G^{2}(t)]/2-\text{Var}[G(t)]\big|_{\ket{\psi(t)}}\nonumber \\
\ge & \left[\frac{g_{\max}(t)-g_{\min}(t)}{2}\right]^{2}+\sum_{k\neq\max,\,\min}\frac{g_{k}^{2}(t)}{2}\ge0,
\end{align}
which concludes the proof. 
\end{proof}
According to Eq.~(\ref{eq:T-free}), the distance of the minimum-time
trajectory for the free evolution is also the minimum-distance trajectory.
In addition, Theorem~\ref{thm:speed} indicates that the speed in
the free evolution is maximum.  Therefore we have the following corollary:
\begin{cor}
\label{corr:minimum}Free evolution generates the global minimum-time
trajectory among all the time-extremal trajectories. 
\end{cor}

Given a set of restricted operators $\{\mathcal{X}_{j}\}$, depending
on the initial and final states, it may occur that the optimal Hamiltonian
for free-evolution given by (\ref{eq:H-free}) is still a legitimate
optimal Hamiltonian that does not contain the disallowed operators
$\{\mathcal{X}_{j}\}$, which corresponds to the solution of Lagrangian
multipliers $\lambda_{j}(t)=0$ for all $j$. On the other hand, according
to Corollary~\ref{corr:minimum}, the free evolution is the global
minimum-time trajectory. In this case, the restricted operators are
not really in effect and  the dynamics of the evolution is then 
trivially restricted. For the dynamics to be non-trivially restricted,
we have the following theorem:
\begin{thm}
\label{thm:non-trivial}The extremal evolution is non-trivially restricted
by the set $\{\mathcal{X}_{j}\}$ if and only if there exists at least
one operator $\mathcal{X}_{j}$ in the restricted set such that 
\begin{equation}
\text{Im}\left[\braket{\psi_{f}^{\perp}\big|\mathcal{X}_{j}\big|\psi_{i}}e^{-\text{i}\phi}\right]\neq0,\label{eq:non-free}
\end{equation}
\end{thm}

\begin{proof}
To exclude the case of $\lambda_{j}(t)=0$ for all $j$, there must
be some operator $\mathcal{X}_{j}$ such that $\text{Tr}[H_{\text{F}}\mathcal{X}_{j}]\neq0$.
Otherwise, the optimal Hamiltonian is the one for the free evolution,
i.e., Eq.~(\ref{eq:H-free}). To see this, we note that
\begin{equation}
\text{Tr}[H_{\text{F}}\mathcal{X}_{j}]=\text{i}(\braket{\psi_{f}^{\perp}\big|\mathcal{X}_{j}\big|\psi_{i}}e^{-\text{i}\phi}-\text{h.c.}).
\end{equation}
Thus, to exclude the case of free evolution, Eq.~(\ref{eq:non-free})
must be satisfied.
\end{proof}

\section{\label{sec:examples}A class of analytically solvable examples for
restricted evolution}

In this section, we consider an important class of solvable examples
of the QB problem where the $\{\mathcal{X}_{j}\}$ forms a closed subalgebra.
The results are summarized in the following theorem:
\begin{thm}
\label{thm:new-class}If the restricted operators form a closed Lie
subalgebra of $su(N)$, i.e.,

\begin{equation}
\mathscr{X}_{kl}\in\text{span}\{\mathcal{X}_{j}\},\,\forall k,\,l,\label{eq:close-algebra}
\end{equation}
 for all the time-extremal trajectories the Lagrange multipliers
are time-independent.  The  optimal Hamiltonian and the unitary
evolution operators can then be expressed as
\begin{align}
H_{\text{}}(t) & =e^{\text{i}Gt}H(0)e^{-\text{i}Gt},\label{eq:H-closed}\\
U_{\text{}}(t) & =e^{\text{i}Gt}e^{-\text{i}[H(0)+G]t},\label{eq:U-closed}
\end{align}
where $\ket{\psi_{f}^{\perp}}$ and $\phi$ are defined in Eqs.~(\ref{eq:psif-perp})
and ~(\ref{eq:phi-def}), respectively, and where $G$ is defined in Eq.~(\ref{eq:G-def})
but is independent of time.
\end{thm}

\begin{proof}
Since the optimal $H(t)$ is the linear combination of the basis $\mathcal{Y}_{j}$,
Eq.~(\ref{eq:close-algebra}) implies that  
\begin{equation}
\text{Tr}[H(t)\mathscr{X}_{kl}]=0,\label{eq:Tr-Ht-calX}
\end{equation}
whence it follows  that $\eta_{jl}(t)=0$. According to
Eqs.~(\ref{eq:lam-0})-(\ref{eq:lam-j}), the dynamics of the Lagrangian
multipliers becomes trivial as they are constant in time $\lambda_{j}(t)=\lambda_{j}(0)$
and $G(t)$ is time-independent. Therefore 
\begin{equation}
\frac{1}{\lambda_{0}}\sum_{j\ge1}\lambda_{j}\tilde{\mathcal{X}}_{j}(t)=V(t)\left(\frac{1}{\lambda_{0}}\sum_{j\ge1}\lambda_{j}\mathcal{X}_{j}\right)V^{\dagger}(t)=V(t)GV^{\dagger}(t)=G.
\end{equation}
In this case Eqs.~(\ref{eq:H-expression})-(\ref{eq:U-expression})
simplify and one finds
\begin{equation}
\text{Tr}[\dot{H}(t)\mathscr{X}_{kl}]=\text{i}\text{Tr}\left(H(t)[\mathscr{X}_{kl},\,G]\right)=0,\label{eq:derivatives}
\end{equation}
where we have used $\dot{\tilde{\mathcal{O}}}(t)=\text{i}[G,\,\tilde{\mathcal{O}}(t)]$
and $[\mathscr{X}_{kl},\,G]\in\text{span}\{\mathcal{X}_{j}\}$ due
to the closure of $\text{span}\{\mathcal{X}_{j}\}$. By mathematical
induction, it follows that
\begin{equation}
\text{Tr}\left[\frac{d^{n}H(t)}{dt^{n}}\mathscr{X}_{kl}\right]=\text{i}^{n}\text{Tr}\left(H(t)[\cdots[[\mathscr{X}_{kl},\,G],\,G],\,\cdots G]\right)=0,\,\forall n.\label{eq:n-derivatives}
\end{equation}
Eq.~(\ref{eq:n-derivatives}) provides a further consistent check
of Eq.~(\ref{eq:Tr-Ht-calX}): as long as one chooses $\text{Tr}[H(0)\mathcal{X}_{kl}]=0$,
thanks to Eq.~(\ref{eq:n-derivatives}) and the continuity of $H(t)$,
Eq.~(\ref{eq:Tr-Ht-calX}) always holds at later times. 
\end{proof}
A few comments are in order. First, when $\text{span}\{\mathcal{X}_{j}\}$
does not form a closed subalgebra, the solution can be complicated
as $G(t)$ will become time-dependent, which presents some analytical
difficulty in solving the Schr\"odinger-like equation~(\ref{eq:V-SE}).
Second, as  mentioned in Sec.~\ref{sec:governing-solution},
solving the dynamics of the Lagrangian multiplier in stage (i) is then
difficult. Theorem~\ref{thm:new-class} specifies the new class of examples
in which the Lagrangian multiplier can be trivially found, i.e., as constants.
Their values together with $H(0)$ in Eqs.~(\ref{eq:H-closed})-(\ref{eq:U-closed})
are determined in stage (ii) in Sec.~\ref{sec:governing-solution}. 

Next, when $U(t)$ takes the form of Eq.~(\ref{eq:U-closed}),
Eq.~(\ref{eq:BConstraint-initial}) becomes
\begin{equation}
\braket{\psi_{i}\big|[H(0)+G,\,U^{\dagger}(T)GU(T)]\big|\psi_{i}}=0,\label{eq:BC-closed}
\end{equation}
which is, like Eqs.~(\ref{eq:IM-BC})-(\ref{eq:BConstraint-initial}),
invariant under renormalization of the Lagrange multipliers. 

In stage (ii), Eq.~(\ref{eq:CHKO-initial})  introduces the constraints
between $\mu_{j}$ and $\lambda_{j}$. In principle, one can solve for
the set of $\lambda_{j}$ in terms of the set of $\mu_{j}$. We refer to the Lagrangian
multipliers that are independent (dependent) of $\mu_{j}$ with $j\in Y$
as \textit{free} (\textit{constrained}). Practically, when the number
of restricted operators is large, the number of constant Lagrangian
multipliers can be also large, which may make the calculation tedious.
Theorem \ref{thm:free-Lag} below indicates that under the condition that the
boundary constraints~(\ref{eq:final-BC}), (\ref{eq:BC-closed}) are
preserved, one can always set the ``free'' Lagrangian to zero,
which reduces the calculation dramatically. 

\begin{thm}
\label{thm:free-Lag}Setting the free Lagrange multipliers in Eqs.~(\ref{eq:H-closed})-(\ref{eq:U-closed}),
i.e., $\lambda_{j}$'s that are independent of $\mu_{j}$ with $j\in Y$,
to be zero will not increase the global minimum-time, as long as the
boundary constraints~(\ref{eq:final-BC}), (\ref{eq:BC-closed}) is
preserved.
\end{thm}

\begin{proof}
The proof of the theorem is rather straightforward: Setting the free
Lagrange multipliers to zero is essentially equivalent to removing
the corresponding constraints. This can be easily shown by noting
that $\lambda_{j}^{\text{free}}=0$ implies that the term $\int_{0}^{T}\lambda_{j}^{\text{free}}f_{j}(H(t))dt$
vanishes in the constraint action, which is effectively equivalent
to the case in which the constraint $f_{j}(H(t))=0$ is absent. 

Furthermore, we observe that the global minimum-time trajectory for
the case containing more constraints is also the locally time-extremal
trajectory in the case in which some of the constraints are removed.
Therefore, setting the free Lagrange multipliers to be zero
cannot increase the global minimum of the evolution time. 
\end{proof}
With the same arguments, one can easily deduce an analogous 
corollary for the constrained Lagrange multipliers:
\begin{cor}
\label{corr:constrained-Lag}Setting the constrained Lagrange multipliers
in Eqs.~(\ref{eq:H-closed})-(\ref{eq:U-closed}) be zero, but still
preserves the norm constraint~(\ref{eq:norm-constraint}) (at $t=0$)
and the boundary constraints~(\ref{eq:final-BC}), (\ref{eq:BC-closed}),
will not increase the globally minimum-time.
\end{cor}

We next illustrate the application of these theorems in representative
examples.

\subsection{$M=1$}

In the case of $M=1$, we find $\mathscr{X}_{11}=0$, which of course
forms a trivial subalgebra. So the dynamics of the Lagrangian multiplier
can be trivially solved, i.e., $\lambda_{0}(t)=\lambda_{0}(0)\equiv\lambda_{0}$
and $\lambda_{1}(t)=\lambda_{1}(0)\equiv\lambda_{1}$. Eqs.~(\ref{eq:H-closed})-(\ref{eq:U-closed})
become
\begin{align}
H(t) & =\exp\left[\text{i}\lambda_{1}\mathcal{X}_{1}t\right]H(0)\exp\left[-\text{i}\lambda_{1}\mathcal{X}_{1}t\right],\label{eq:M1-H}\\
U(t) & =\exp\left[\text{i}\lambda_{1}\mathcal{X}_{1}t\right]\exp\left[-\text{i}\left(H(0)+\lambda_{1}\mathcal{X}_{1}\right)t\right].\label{eq:M1-U}
\end{align}
Without loss of generality one can choose $\lambda_{0}=1$. As we
have mentioned, thanks to Eqs.~(\ref{eq:sym-a})-(\ref{eq:sym-b}),
$\lambda_{1}$ and $\lambda_{0}$ can be  renormalized to $\tilde{\lambda_{1}}$
and $\tilde{\lambda}_{0}$, respectively, in order to satisfy Eq.~(\ref{eq:Re-BC})
without changing $H(t)$ and $U(t)$. $\tilde{F}(t)$ is computed
from the normalized Lagrange multipliers $\tilde{\lambda}_{j}(t)$
while $F(t)$ is computed from the unnormalized Lagrange multipliers
$\lambda_{j}(t)$. 

Eqs.~(\ref{eq:M1-H})-(\ref{eq:M1-U}) generalize the restricted
in example for two-level system by CHKO~\citep{carlini2006timeoptimal}
to $N$-level systems. One can take $\mathcal{X}_{1}=\sigma_{z}$,
which is the example presented in~\citep{carlini2006timeoptimal}.
Following the Step $1$ in the recipe in Sec.~\ref{sec:governing-solution},
we consider 
\begin{align}
H(0) & =\mu_{x}\sigma_{x}+\mu_{y}\sigma_{y},\\
F(0) & =\mu_{x}\sigma_{x}+\mu_{y}\sigma_{y}+\lambda_{1}\sigma_{z}.
\end{align}
For the initial state $\ket{\psi_{i}}=\ket{e_{1}}=\ket{+x}$
and $\ket{e_{2}}=\ket{\psi_{f}^{\perp}}=\ket{-x}$,  Eq.~(\ref{eq:CHKO-initial})
implies that 
\begin{align}
\braket{e_{1}\big|F(0)\big|e_{1}} & =-\braket{e_{2}\big|F(0)\big|e_{2}}=\mu_{x}=0.
\end{align}
Eq.~(\ref{eq:mu-j-norm}) implies that $\mu_{y}=\omega$. Thus Eqs.~(\ref{eq:M1-H})-(\ref{eq:M1-U})
become 
\begin{align}
H(t) & =\omega\exp\left[\text{i}\lambda_{1}\sigma_{z}t\right]\sigma_{y}\exp\left[-\text{i}\lambda_{1}\sigma_{z}t\right],\label{eq:M1-H-2level}\\
U(t) & =\exp\left[\text{i}\lambda_{1}\sigma_{z}t\right]\exp\left[-\text{i}\left(\omega\sigma_{y}+\lambda_{1}\sigma_{z}\right)t\right].\label{eq:M1-U-2level}
\end{align}
 Reference~\citep{carlini2006timeoptimal} assumes
that after the renormalization of $\lambda_{1}$ and $\lambda_{0}$,
all the pairs of $(\tilde{\lambda}_{1},\,T)$ would give rise to a
locally time-extremal trajectory between the initial state $\ket{\psi_{i}}$
and $U(T)\ket{\psi_{i}}$ as Eq.~(\ref{eq:IM-BC}) is not taken into
account. As one can see from Fig.~\ref{fig:M1}(b), there are pairs
of $(\lambda_{1},\,T)$ that violate Eq.~(\ref{eq:IM-BC}). These
trajectories, that satisfy the CHKO equation~(\ref{eq:CHKO})-(\ref{eq:CHKO-initial}),
the constraints~(\ref{eq:traceless})-(\ref{eq:term-constraint}) and
the boundary condition~(\ref{eq:final-BC}),  are not, even locally,
extremal trajectories. This aspect was ignored in Ref.~\citep{carlini2006timeoptimal}
for this simple case with $N=2$.
\begin{figure}
\begin{centering}
\includegraphics[scale=0.23]{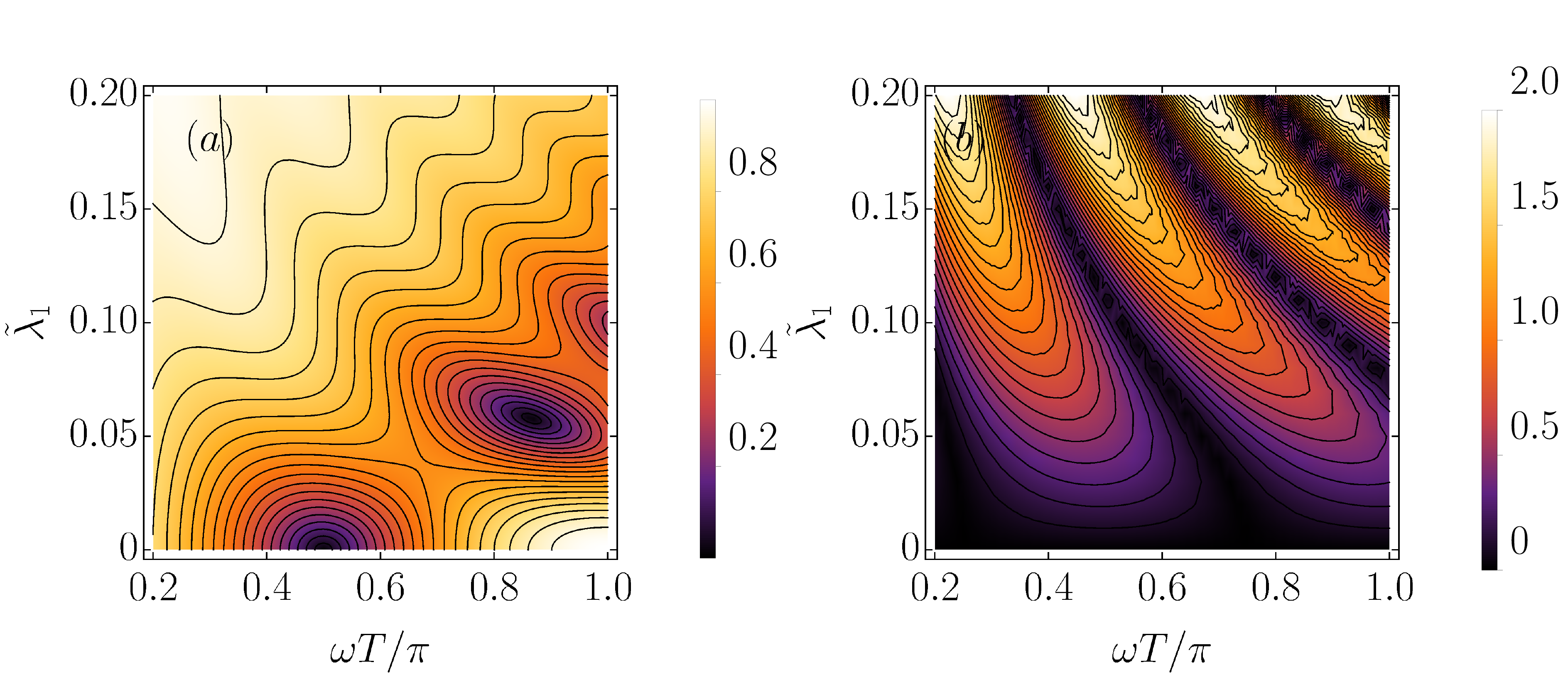}
\par\end{centering}
\caption{\label{fig:M1}Figs. (a) and (b) show the contour plots of amplitude
$|\braket{\psi_{f}\big|\psi_{i}}|$ and $\text{Im}\braket{\psi_{f}\big|H(T)\tilde{F}(T)\big|\psi_{f}}$
versus the values of $\tilde{\lambda}_{1}$ and $T$, respectively.
Value of parameter $\omega=10$. In this particular case, we note
$\text{Re}\braket{\psi_{f}\big|H(T)\tilde{F}(T)\big|\psi_{f}}=\omega^{2}$
so that $\tilde{\lambda}_{j}=\lambda_{j}/\omega^{2}$. At each point,
$\tilde{\lambda}_{1}=\lambda_{1}/\text{Re}\braket{\psi_{f}\big|H(T)F(T)\big|\psi_{f}}$.
Only pairs of $(\tilde{\lambda}_{1},\,T)$ in the black regions on
(b) satisfy Eq.~(\ref{eq:IM-BC}) and therefore represents a locally
time-extremal trajectory. For any other pairs of $(\tilde{\lambda}_{1},\,T)$
in other regions of (b), although they satisfy all the equations of
QB problem except Eq.~(\ref{eq:IM-BC}), they do not represent the
time-extremal trajectory. }
\end{figure}

One can readily calculate that in the basis of $\{\ket{0},\,\ket{1}\}$,
\begin{equation}
\ket{\psi(T)}=\frac{1}{\sqrt{2}}\begin{pmatrix}e^{\text{i}\lambda_{1}T}\cos(\Lambda_{1}T)-e^{\text{i}\lambda_{1}T}e^{\text{i}\theta_{1}}\sin(\Lambda_{1}T)\\
e^{-\text{i}\lambda_{1}T}\cos(\Lambda_{1}T)+e^{\text{i}\theta_{1}-\text{i}\lambda_{1}T}\sin(\Lambda_{1}T)
\end{pmatrix},
\end{equation}
where 
\begin{equation}
\Lambda_{1}\equiv\sqrt{\lambda_{1}^{2}+\omega^{2}},\,\cos\theta_{1}\equiv\omega/\Lambda_{1},\,\sin\theta_{1}=\lambda_{1}/\Lambda_{1}.\label{eq:Lambda-sin-def}
\end{equation}
A straightforward calculation of the l.h.s. of Eq.~(\ref{eq:BC-closed})
implies that 
\begin{equation}
\lambda_{1}\cos(2\Lambda_{1}T)=0.\label{eq:BC-1qubit}
\end{equation}

According to Theorem~\ref{thm:non-trivial},  in this case
$\braket{\psi_{f}^{\perp}\big|\mathcal{X}_{j}\big|\psi_{i}}=\braket{-x\big|\sigma_{z}\big|+x}=1$.
Thus, if $\phi=2m\pi,\,\pi+2m\pi$ with $m\in \mathbb{Z}$, the constraint has no effect.
This is consistent with the fact that setting $\lambda_{1}=0$ and
equating $\ket{\psi(T)}\sim\ket{\psi_{f}}$ will lead to $\phi=2m\pi,\,\pi+2m\pi$.
In this case, Eq.~(\ref{eq:BC-1qubit}) is trivially satisfied as
one can see from Fig.~\ref{fig:M1}~(b). The genuine restricted
evolution with $\lambda_{1}\neq0$ has not been discussed before. Yet, 
this is the simplest case where there is only one restricted operator.
When we have freedom to choose the final state $\ket{\psi_{f}}$,
any choice of $\lambda_{1}\neq0$ such that $\cos(2\Lambda_{1}T)=0$
will generate a local-time extremal trajectory between $\ket{\psi_{i}}$
and $\ket{\psi_{f}}=\ket{\psi(T)}$.

What if the final state is a priori
known? With Eq.~(\ref{eq:psi-ref}), in the basis of $\{\ket{0},\,\ket{1}\}$,
the final state can be parameterized as follows:
\begin{equation}
\ket{\psi_{f}}=\frac{1}{\sqrt{2}}\begin{pmatrix}e^{\text{i}\phi}\cos\Omega_{\text{B}}+1\sin\Omega_{\text{B}}\\
e^{\text{i}\phi}\cos\Omega_{\text{B}}-\sin\Omega_{\text{B}}
\end{pmatrix}.
\end{equation}
Now we would like to satisfy both Eq.~(\ref{eq:BC-closed}) and Eq.~(\ref{eq:final-BC}).
Equation~(\ref{eq:BC-1qubit}) implies that
\begin{equation}
\Lambda_{1}T=\frac{\pi}{4}+\frac{k\pi}{2},\,k\in\mathbb{N}.\label{eq:Lambda1T}
\end{equation}
To satisfy Eq.~(\ref{eq:final-BC}), we would like to have $\rho(T)=\rho_{f}$,
where $\rho(T)=\ket{\psi(T)}\bra{\psi(T)}$ and $\rho_{f}=\ket{\psi_{f}}\bra{\psi_{f}}$.
Using Eq.~(\ref{eq:Lambda1T}), this leads to 
\begin{align}
(-1)^{k}\cos\theta_{1} & =-\cos\phi\sin(2\Omega_{\text{B}}),\label{eq:diag}\\
(-1)^{k}\sin\theta_{1}\sin(2\lambda_{1}T) & =\cos(2\Omega_{\text{B}}),\label{eq:off-diag1}\\
(-1)^{k}\sin\theta_{1}\cos(2\lambda_{1}T) & =\sin\phi\sin(2\Omega_{\text{B}}).\label{eq:off-diag2}
\end{align}
Our goal now is to find a solution for $\lambda_{1}$ and $T$ so that
Eqs.~(\ref{eq:Lambda1T}-\ref{eq:off-diag2}) are satisfied consistently.
Taking the ratio between Eq.~(\ref{eq:off-diag1}) and Eq.~(\ref{eq:off-diag2}),
we find 
\begin{equation}
\lambda_{1}T=\frac{1}{2}\left[\text{arccot}\left(\sin\phi\tan(2\Omega_{\text{B}})\right)+l\pi\right],\,l\in\mathbb{Z}.\label{eq:lam1T}
\end{equation}
Substituting Eq.~(\ref{eq:lam1T}) into Eqs.~(\ref{eq:Lambda-sin-def}), (\ref{eq:Lambda1T}),
we find
\begin{align}
T(k,\,l) & =\frac{1}{|\omega|}\sqrt{\left(\frac{\pi}{4}+\frac{k\pi}{2}\right)^{2}-\frac{1}{4}\left(\text{arccot}\left[\sin\phi\tan(2\Omega_{\text{B}})\right]+l\pi\right)^{2}},\label{eq:T} \\
\sin\theta_{1} & =\frac{\text{arccot}\left[\sin\phi\tan(2\Omega_{\text{B}})\right]/2+l\pi/2}{\pi/4+k\pi/2}.
\end{align}
Clearly, $k\in\mathbb{N}$ and $l\in\mathbb{Z}$ must satisfy the
constraint that the expression under the square root on the r.h.s. of Eq.~\eqref{eq:T} is positive. Finally, we note that with Eq.~(\ref{eq:Lambda1T}),
Eq.~(\ref{eq:off-diag1}) implies Eq.~(\ref{eq:off-diag2}) and
Eq.~(\ref{eq:diag}) or vice versa. Therefore, the pair $(k,\,l)$
also needs to satisfy Eq.~(\ref{eq:off-diag1}), which leads to the
following compatibility condition:
\begin{equation}
\frac{2\left(\text{arccot}\left[\sin\phi\tan(2\Omega_{\text{B}})\right]/\pi+l\right)}{(1+2k)\sqrt{1+[\cot(2\Omega_{\text{B}})]^{2}(\text{csc}\phi)^{2}}}=\sin(2\Omega_{\text{B}})\sin\phi,
\label{eq:Compatibility}
\end{equation}
where $\phi \neq 2m\pi,\,\pi+2m\pi$ with $m\in\mathbb{Z}$ since the free
evolution is excluded. The global minimum time can be obtained by finding 
$\min_{k,\,l}T(k,\,l)$ such that the pair $(k,\,l)$ satisfies Eq.~(\ref{eq:Compatibility})
with $k\in\mathbb{N}$, $l\in\mathbb{Z}$. 

\subsection{\label{subsec:multiple-constraints}A two-qubit example with multiple
constraints }

In this section, we consider an example in which the set of restricted operators
contain more than one operators and form a Lie sub algebra. For the
sake of simplicity, we shall index the Pauli operators by number $1,2,3$
rather than $x,\,y,\,z$. Let us consider a two-qubit example, where
\begin{equation}
\{\mathcal{X}_{j}\}\in\{\sigma_{1}^{1}\sigma_{2}^{\alpha},\,\sigma_{1}^{\alpha}\sigma_{2}^{1},\,\sigma_{i}^{\alpha}\},\,\alpha=1,\,2,3,\,i=1,\,2.
\end{equation}
In this case,  both single-qubit operations and any two-qubit operations involving
$\sigma_{i}^{1}$($\sigma_{i}^{x}$) operation on any of the qubits
are forbidden. Explicit computation yields
\begin{align}
[\sigma_{1}^{1}\sigma_{2}^{\alpha},\,\sigma_{1}^{\beta}\sigma_{2}^{1}] & =0,\,\,\alpha,\,\beta=2,\,3,\\
\text{i}[\sigma_{1}^{1}\sigma_{2}^{2},\,\sigma_{1}^{1}\sigma_{2}^{3}] & =-2\sigma_{2}^{1},\\
\text{i}[\sigma_{1}^{3}\sigma_{2}^{1},\,\sigma_{1}^{2}\sigma_{2}^{1}] & =2\sigma_{1}^{1},\\
\text{i}[\sigma_{1}^{1}\sigma_{2}^{\alpha},\,\sigma_{1}^{1}\sigma_{2}^{1}] & =\text{2}\epsilon_{\alpha\beta1}\sigma_{2}^{\beta},\,\alpha,\,\beta=2,\,3,\\
\text{i}[\sigma_{1}^{\alpha}\sigma_{2}^{1},\,\sigma_{1}^{1}\sigma_{2}^{1}] & =\text{2}\epsilon_{\alpha\beta1}\sigma_{1}^{\beta},\,\alpha,\,\beta=2,\,3.
\end{align}
Thus,  $\text{span}\{\mathcal{X}_{j}\}$ forms a closed subalgebra
of $su(4)$ and the Lagrange multipliers are constants.
We simply denote $\lambda_{j}\equiv\lambda_{j}(0)$. Once again, we
can choose $\lambda_{0}=1$ for the ease of calculation. As we have
mentioned previously, one can always renormalize $\lambda_{j}$'s
according to Eqs.~(\ref{eq:sym-a})-(\ref{eq:sym-b}) such that (\ref{eq:Re-BC})
is satisfied while keeping $H(t)$ and $U(t)$ unchanged. Theorem~\ref{thm:new-class}
gives then  the optimal Hamiltonian  
\begin{align}
H(t) & =\sum_{(\alpha\beta)\in Y}\mu_{\alpha\beta}e^{\text{i}\sum_{(\gamma\delta)\in X}\lambda_{\gamma\delta}\sigma_{1}^{\gamma}\sigma_{2}^{\beta}t}\sigma_{1}^{\alpha}\sigma_{2}^{\beta}e^{-\text{i}\sum_{(\gamma\delta)\in X}\lambda_{\gamma\delta}\sigma_{1}^{\gamma}\sigma_{2}^{\beta}t},\label{eq:H-2qu}\\
U(t) & =e^{\text{i}\sum_{(\gamma\delta)\in X}\lambda_{\gamma\delta}\sigma_{1}^{\gamma}\sigma_{2}^{\beta}t}e^{-\text{i}[\sum_{(\alpha\beta)\in Y}\mu_{\alpha\beta}\sigma_{1}^{\alpha}\sigma_{2}^{\beta}+\sum_{(\gamma\delta)\in X}\lambda_{\gamma\delta}\sigma_{1}^{\gamma}\sigma_{2}^{\beta}]t},\label{eq:U-2qu}
\end{align}
where 
\begin{align}
X & =\{10,\,20,\,30,\,01,\,01,\,02,\,03,\\
 & 11,\,12,\,21,\,13,\,31\},\nonumber \\
Y & =\{22,\,33,\,23,\,32\}.
\end{align}

As we have mentioned, due to Eq.~(\ref{eq:CHKO-initial}), $\lambda_{\gamma\delta}$
and $\mu_{\alpha\beta}$ are not independent. For example, if
we start with the ground state $\ket{e_{1}}=\ket{\psi_{i}}=\ket{11}$
and $\ket{e_{2}}=\ket{\psi_{f}^{\perp}}=\ket{00}$ so that the final
state is 
\begin{equation}
\ket{\psi_{f}}=\cos\Omega_{\text{B}}e^{\text{i}\phi}\ket{00}+\sin\Omega_{\text{B}}\ket{11}.\label{eq:2qu-psif}
\end{equation}
Among all the restricted operators, we already find
\begin{equation}
\braket{11\big|\sigma_{1}^{1}\sigma_{2}^{1}\big|00}=1,\,\braket{11\big|\sigma_{1}^{1}\sigma_{2}^{2}\big|00}=\braket{11\big|\sigma_{1}^{2}\sigma_{2}^{1}\big|00}=\text{i}.
\end{equation}
According to Theorem~\ref{thm:non-trivial}, regardless the values
of $\phi$, the optimal time-evolution is restricted. We choose
the remaining basis as $\ket{e_{3}}=\ket{10}$ and $\ket{e_{4}}=\ket{01}$. 

Furthermore, we show in Appendix~\ref{sec:Initial-form} that Eq.~(\ref{eq:CHKO-initial})
leads to the following constraints among the coefficients.
\begin{align}
\mu_{33} & =0,\, & \lambda_{30}=\lambda_{03} & =0,\label{eq:2qu-coeff1}\\
\mu_{23}+\lambda_{20} & =0,\, & \lambda_{13}+\lambda_{10} & =0,\label{eq:2qu-coeff2}\\
\mu_{32}+\lambda_{02} & =0,\, & \lambda_{31}+\lambda_{01} & =0,\label{eq:2qu-coeff4}\\
\mu_{22}+\lambda_{11} & =0,\, & \lambda_{12}-\lambda_{21} & =0.\label{eq:2qu-coeff6}
\end{align}
Eq.~(\ref{eq:mu-j-norm}) indicates that 
\begin{equation}
\mu_{23}^{2}+\mu_{32}^{2}+\mu_{22}^{2}=\frac{\omega^{2}}{2}.
\end{equation}

Thus, the free parameters are $\lambda_{13}$ $\lambda_{31}$,
$\lambda_{12}$ and $\lambda_{21}$. According
to Theorem~\ref{thm:free-Lag}, they can be set to be zero, provided
the boundary constraints~(\ref{eq:final-BC}),~(\ref{eq:BC-closed})
are satisfied, which will be manifestly true later but is assumed
for now. The calculation can be further reduced. Among the
three Lagrange multipliers left, we choose $\lambda_{11}\neq0$ with
$\lambda_{20}=\lambda_{02}=0$. If such a choice can satisfy boundary
constraints~(\ref{eq:final-BC}),~(\ref{eq:BC-closed}), according
to Corollary~\ref{corr:constrained-Lag}, it provides the
globally minimum-time trajectory. The reason we do not set $\lambda_{11}=0$
is that it will result in a local Hamiltonian  in the qubits that
cannot general nonlocal evolution. Therefore, an entangled state cannot
be generated and Eq.~(\ref{eq:final-BC}) is violated. With this choice,
Eqs.~(\ref{eq:H-2qu})-(\ref{eq:U-2qu}) become
\begin{align}
H & =\frac{\omega}{\sqrt{2}}\sigma_{1}^{2}\sigma_{2}^{2},\\
U(t) & =\exp\left[-\frac{\text{i}\omega t}{\sqrt{2}}\sigma_{1}^{2}\sigma_{2}^{2}\right],
\end{align}
with 
\begin{equation}
G=-\frac{\omega}{\sqrt{2}}\sigma_{1}^{1}\sigma_{2}^{1},
\end{equation}
where we have used the fact that $[\sigma_{1}^{1}\sigma_{2}^{1},\,\sigma_{1}^{2}\sigma_{2}^{2}]=0$.
To satisfy that the condition $\ket{\psi_{f}^{\perp}}=\ket{00}$,
we must have
\begin{align}
\braket{e_{3}\big|U(T)\big|\psi_{i}} & =0,\label{eq:e3-comp}\\
\braket{e_{4}\big|U(T)\big|\psi_{i}} & =0.\label{eq:e4-comp}
\end{align}
A straightforward application of the Baker-Campell-Hausdorff formula
shows that Eqs.~(\ref{eq:e3-comp})-(\ref{eq:e4-comp}) is satisfied. The time evolution reads 
\begin{equation}
U(T)\ket{\psi_{i}}=\cos(\omega T/\sqrt{2})\ket{11}+\text{i}\sin(\omega T/\sqrt{2})\ket{00}.
\end{equation}
In this case, Eq.~(\ref{eq:BC-closed}) is trivially satisfied.
Upon setting $\ket{\psi_{f}}\sim U(T)\ket{\psi_{i}}$, which leads
to 
\[
\phi=-\frac{\pi}{2}+2k\pi,\,k\in\mathbb{Z},
\]
the boundary condition~(\ref{eq:final-BC}) is also  satisfied.
Therefore, we find the globally minimum-time evolution is 
\begin{equation}
T=\frac{\sqrt{2}\Omega_{\text{B}}}{\omega}.
\end{equation}
In comparison to the free evolution, the constraints make the time-optimal
evolution longer by a factor $\sqrt{2}$.

Similar calculations can be done for other initial and final states,
e.g., $\ket{\psi_{i}}=\ket{01}$ and $\ket{\psi_{f}^{\perp}}=\ket{10}$,
$\ket{\psi_{i}}=\ket{+x,\,+x}$ and $\ket{\psi_{f}^{\perp}}=\ket{-x,\,-x}$
, $\ket{\psi_{i}}=\ket{+x,\,-x}$ and $\ket{\psi_{f}^{\perp}}=\ket{-x,\,+x}$,
etc. We would like to remark that although we only give two analytic examples,
the class where restricted operators form
a closed Lie algebra is rich. In particular, we expect  this class
to contain analytically solvable instances of the QB problem in many-body restricted
Hamiltonians, of relevance to the study of quantum
speed limits in many-body quantum systems~\citep{bukov2019geometric,delcampo2021probing}. 

\section{\label{sec:conclusion}Conclusion}

In summary, we argued that unlike the classical brachistochrone
problem, the final boundary condition in the QB problem should be
considered as movable according to the $U(1)$ gauge transformation
and the QB should be solved by variational calculus with movable boundary
conditions. The effect of the movable endpoint introduces an additional
constraint, unrecognized in the original formulation of QB by CHKO~\citep{carlini2006timeoptimal} and ensuing literature.
Furthermore, we have also provided an alternative derivation of the
QB equations based on the proper observation of the boundary condition.
An advantage of the current approach is that it requires much less effort than the original derivation
by CHKO. 

Using it, we have reported a general expression for the optimal Hamiltonian and optimal
unitary evolution operator and derived the governing equation for
the dynamics of the Lagrange multipliers in the QB problem. We have also proposed
a numerical algorithm that generates time-extremal trajectories,
taking into account the additional constraint at the final time.
Furthermore, we have identified an important class of analytically solvable
examples of the QB problem where the restricted operators form a closed
Lie algebra. In this case, the Lagrange multipliers become constants
and the optimal Hamiltonian and optimal unitary operators take a simple
form. This opens up the possibility to study QB in many-body systems.
We have illustrated with specific examples that the effect of the moving
endpoint cannot be ignored in general. Indeed, doing so can lead to an erroneous
identification of the time-extremal trajectories. 

Our results here open the door to investigate the geometry of the evolution of many-body
quantum systems. Many questions are open based on our results here, such as
the combination of  the recipe for generating time-extremal trajectories
proposed here with other algorithms to develop a full numerical framework to
solve the CHKO equation, the study of the QB in many-body
quantum systems combined with many-body techniques, and the application of the analytical
findings reported here to the optimal generation of a target quantum gate. 

\section{Acknowledgement}

It is a pleasure to acknowledge discussions with Niklas H\"ornedal,
Aritra Kundu, Shengshi Pang, Kazutaka Takahashi and Hongzhe Zhou.

\appendix
\begin{widetext}

\section{\label{sec:equivalence}The equivalence between Eq.~(\ref{eq:CHKO-psi})
and Eq.~(\ref{eq:CHKO})}

The solution to Eq.~(\ref{eq:CHKO}) is 
\begin{equation}
F(t)=U(t)F(0)U^{\dagger}(t),
\end{equation}
where $U(t)$ is the unitary evolution generated by $H(t)$. One can
easily check that $F(t)$  satisfies Eq.~(\ref{eq:CHKO-psi}). Let us proof
the converse. Obviously, the second equation of Eq.~(\ref{eq:CHKO-psi})
implies the initial condition in Eq.~(\ref{eq:CHKO}). Our goal now
is to prove the differential equation Eq.~(\ref{eq:CHKO}). Equation~(\ref{eq:CHKO-psi})
is equivalent to
\begin{align}
\left\{ \dot{F}(t)+\text{i}[H(t),\,F(t)]\right\} \mathcal{P}(t) & =0,\label{eq:CHKOpsi-Pt}\\
U^{\dagger}(t)\{F(t),\,\mathcal{P}(t)\}U(t) & =\bar{F}(t),\label{eq:CHKOH-Fbar}
\end{align}
where
\begin{equation}
\bar{F}(t)\equiv U^{\dagger}(t)F(t)U(t),\label{eq:Fbar-def}
\end{equation}
and $\mathcal{P}(t)=U(t)\mathcal{P}(0)U(t)$ with $\mathcal{P}(0)=\ket{\psi_{i}}\bra{\psi_{i}}$. 

Eq.~(\ref{eq:CHKOpsi-Pt}) can be further rewritten as 
\begin{equation}
U^{\dagger}(t)\left\{ \dot{F}(t)+\text{i}[H(t),\,F(t)]\right\} U(t)\mathcal{P}(0)=0.\label{eq:CHKOpsi-P0}
\end{equation}
According to Eq.~(\ref{eq:Fbar-def}), it can be readily obtained
\begin{equation}
\dot{\bar{F}}(t)=\text{i}U^{\dagger}(t)[H(t),\,F(t)]U(t)+U^{\dagger}(t)\dot{F}(t)U(t).\label{eq:Fbar-dot}
\end{equation}
Substituting Eq.~(\ref{eq:Fbar-dot}) into Eq.~(\ref{eq:CHKOpsi-P0}) yields 
\begin{equation}
\dot{\bar{F}}(t)\mathcal{P}(0)=0,
\end{equation}
from which we conclude that $\bar{F}(t)\mathcal{P}(0)$ is a constant of
motion, that is 
\begin{equation}
\bar{F}(t)\mathcal{P}(0)=\bar{F}(0)\mathcal{P}(0).\label{eq:CHKOpsi-Fbar-P}
\end{equation}
 Furthermore, we note the relations
\begin{align}
U^{\dagger}(t)F(t)\mathcal{P}(t)U(t) & =\bar{F}(t)\mathcal{P}(0),\\
U^{\dagger}(t)\mathcal{P}(t)F(t)U(t) & =\mathcal{P}(0)\bar{F}(t).
\end{align}
Using them Eq.~(\ref{eq:CHKOH-Fbar}) can be rewritten as 
\begin{equation}
\{\bar{F}(t),\,\mathcal{P}(0)\}=\bar{F}(t).\label{eq:Fbar-P0-anticomm}
\end{equation}
Combining Eq.~(\ref{eq:CHKOpsi-Fbar-P}) with Eq.~(\ref{eq:Fbar-P0-anticomm}),
it follows that 
\begin{equation}
\bar{F}(t)=\bar{F}(0)\mathcal{P}(0)+\mathcal{P}(0)\bar{F}(0),
\end{equation}
which is a constant of motion. As a result, $\dot{\bar{F}}(t)=0$ leads to Eq.~(\ref{eq:CHKO})
in the main text. 

\section{\label{sec:Rederiving-CHKO-gate}Rederiving the CHKO equation for
quantum gate implementation with minimum efforts}

In Ref.~\citep{carlini2007timeoptimal}, to study the QB
equation for quantum gates, CHKO constructed the following action:
\begin{equation}
S_{\text{CHKO}}(\ket{\psi},\,H,\,\ket{\phi},\,\lambda_{j})=\sum_{\alpha=\text{T},\,\text{S},\,\text{C}}\int_{0}^{T}L_{\text{\ensuremath{\alpha} }}dt,\label{eq:S-CHKO-U}
\end{equation}
where the time, Schr\"odinger, and the constraint Lagrangians are
defined as
\begin{align}
L_{\text{T}} & =\frac{\sqrt{g_{tt}}}{v(t)},\label{eq:LT-def-U}\\
L_{\text{\ensuremath{\text{S}}}} & =\text{Tr}\left(\Lambda(t)[\text{i}\dot{U}(t)U^{\dagger}(t)-H(t)]\right),\label{eq:LS-def-U}\\
L_{\text{C}} & =\sum_{j}\lambda_{j}(t)f_{j}(H(t)),\label{eq:LC-def-U}
\end{align}
with the boundary conditions 
\begin{align}
U(0) & =\mathbb{I},\label{eq:U-initial}\\
U(t_{f}) & \sim U_{f}.\label{eq:U-final}
\end{align}
Again, one can always make the Hamiltonian traceless, i.e., $\text{Tr}[H(t)]=0$. Further,
\begin{equation}
g_{tt}\equiv\text{Tr}[\dot{U}^{\dagger}(t)\dot{U}(t)]+\frac{1}{N}\left(\text{Tr}[\dot{U}^{\dagger}(t)U(t)]\right)^{2}\label{eq:gtt-U}
\end{equation}
is the metric on the manifold of quantum unitary matrices induced by the
Hilbert space norm which is invariant under $U(1)$-gauge transformation.
CHKO computed the speed of $v(t)$ by substituting the Schr\"odinger
equation into Eq.~(\ref{eq:gtt-U}) and obtained 
\begin{equation}
v(t)=\sqrt{g_{tt}}=\sqrt{\text{Tr}[H^{2}(t)]}.
\end{equation}
As we have discussed in the main text, to derive the CHKO equation,
there is no need to perform the variational calculus with respect
to $L_{\text{T}}$. One can directly set $\delta S_{\text{T}}=0$.
On the other hand, since $S_{\text{C}}$ is independent of $U$, it can be readily found that
$\delta S_{\text{C}}=0$ under the the variation of $U$. Next, we would like to compute $\delta S_{\text{S}}$. 

For the Lagrangian $L=L(U,\,\dot{U},\,U^{\dagger},\,\dot{U}^{\dagger})$, minimization of the action yields
\begin{align}
\delta\int_{0}^{T}Ldt & =\int_{0}^{T}dt\text{Tr}\left(\frac{\partial L}{\partial U}\delta U+\frac{\partial L}{\partial\dot{U}}\delta\dot{U}+\frac{\partial L}{\partial U^{\dagger}}\delta U^{\dagger}+\frac{\partial L}{\partial\dot{U}^{\dagger}}\delta\dot{U}^{\dagger}\right)\nonumber \\
 & =\int_{0}^{T}dt\text{Tr}\left[\left(\frac{\partial L}{\partial U}-\frac{d}{dt}\frac{\partial L}{\partial\dot{U}}\right)\delta U+\left(\frac{\partial L}{\partial U^{\dagger}}-\frac{d}{dt}\frac{\partial L}{\partial\dot{U}^{\dagger}}\right)\delta U^{\dagger}\right]\nonumber \\
 & +\text{Tr}\left(\frac{\partial L}{\partial\dot{U}}\delta U\right)\bigg|_{t=0}^{t=T}+\text{Tr}\left(\frac{\partial L}{\partial\dot{U}^{\dagger}}\delta U^{\dagger}\right)\bigg|_{t=0}^{t=T}\,.
\end{align}
Note that due to $UU^{\dagger}=\mathbb{I}$, $\delta U^{\dagger}$
and $\delta U$ are related to each other as follows:
\begin{align}
\delta U^{\dagger} & =-U^{\dagger}\delta UU^{\dagger}.\label{eq:delU-dag}
\end{align}
Using Eq.~(\ref{eq:delU-dag}), we thus obtain
\begin{align}
\delta\int_{0}^{T}Ldt & =\int_{0}^{T}dt\text{Tr}\left[\left(\frac{\partial L}{\partial U}-\frac{d}{dt}\frac{\partial L}{\partial\dot{U}}-U^{\dagger}\frac{\partial L}{\partial U^{\dagger}}U^{\dagger}+U^{\dagger}\frac{d}{dt}\frac{\partial L}{\partial\dot{U}^{\dagger}}U^{\dagger}\right)\delta U\right]\nonumber \\
 & +\text{Tr}\left[\left(\frac{\partial L}{\partial\dot{U}}-U^{\dagger}\frac{\partial L}{\partial\dot{U}^{\dagger}}U^{\dagger}\right)\delta U\right]\bigg|_{t=0}^{t=T}\,.\label{eq:delS-general-varU}
\end{align}
Taking $L=L_{\text{S}}$ and applying the fixed boundary condition,
we arrive at

\begin{align}
\delta S_{\text{S}} & =-\text{i}\int_{0}^{T}dt\text{Tr}\left\{ U^{\dagger}(t)\left(\dot{\Lambda}(t)+\Lambda(t)\dot{U}(t)U^{\dagger}(t)+U(t)\dot{U}^{\dagger}(t)\Lambda(t)\right)\delta U(t)\right\} \nonumber \\
 & =-\text{i}\int_{0}^{T}dt\text{Tr}\left\{ U^{\dagger}(t)\left(\dot{\Lambda}(t)+\text{i}[H(t),\,\Lambda(t)]\right)\delta U(t)\right\} ,
\end{align}
where we have used $H(t)=\text{i}\dot{U}(t)U^{\dagger}(t)$. Therefore,
the Euler-Lagrange equation for $U(t)$ is 
\begin{equation}
\dot{\Lambda}(t)+\text{i}[H(t),\,\Lambda(t)]=0.\label{eq:Lambda-dfif}
\end{equation}
Keeping $U$ fixed, one can easily find the identities
\begin{align}
\delta S_{\text{S}} & =-\int_{0}^{T}\text{Tr}\left(\Lambda(t)\delta H(t)\right)dt,\\
\delta S_{\text{C}} & =\int_{0}^{T}\text{Tr}\left(F(t)\delta H(t)\right)dt.
\end{align}
The Euler-Lagrangian equation for $H(t)$ is 
\begin{equation}
\Lambda(t)=F(t).\label{eq:Lambda-F}
\end{equation}
Obviously, Eqs.~(\ref{eq:Lambda-dfif}),~(\ref{eq:Lambda-F}) imply
Eq.~(\ref{eq:CHKO}) in the main text.

\section{\label{sec:FullQB-gate}Derivation of the full quantum brachistochrone
equation for quantum gate implementation}

As with Eq.~(\ref{eq:varpsi}), one can show that
\begin{equation}
\tilde{U}(\tilde{T})-U(T)=\dot{U}(T)\delta T+\delta U(T).
\end{equation}
 The boundary condition~(\ref{eq:U-final}) again dictates that
\[
\tilde{U}(\tilde{T})=e^{\text{i}\delta\theta(T)}U(T),
\]
and we find 
\begin{equation}
\delta U(T)=-\dot{U}(T)\delta T+\text{i}\delta\theta(T)U(T).
\end{equation}
This equation is also discussed in Ref.~\citep{wakamura2020ageneral},
though the issue of the moving boundary is not explicitly mentioned
there. When varying $U(t)$ while keeping $H(t)$ fixed, we find 
\begin{align}
\delta S_{\text{S}} & =\int_{0}^{T+\delta T}L_{\text{S}}(U+\delta U)dt-\int_{0}^{T}L_{\text{S}}(U)dt\nonumber \\
 & =L_{\text{S}}(U)\delta T+\delta\int_{0}^{T}L_{\text{S}}(U+\delta U)dt-\int_{0}^{T}L_{\text{S}}(U)dt\nonumber \\
 & =L_{\text{S}}(U)\delta T+\text{i}\text{Tr}\left[U^{\dagger}(t)\Lambda(t)\delta U(t)\right]\big|_{t=0}^{t=T}-\text{i}\int_{0}^{T}dt\text{Tr}\left\{ U^{\dagger}(t)\left(\dot{\Lambda}(t)+\text{i}[H(t),\,\Lambda(t)]\right)\delta U(t)\right\}\,. 
\end{align}
Therefore, the relation
\begin{equation}
\delta S_{\text{CHKO}}=\delta T+\delta S_{\text{S}}=0
\end{equation}
 yields not only the CHKO equation~(\ref{eq:CHKO}), but also 
\begin{align}
\text{i}\text{Tr}\left[U^{\dagger}(T)\Lambda(T)\dot{U}(T)\right] & =1,\label{eq:U-add1}\\
\text{Tr}\left[U^{\dagger}(T)\Lambda(T)U(T)\right] & =0.\label{eq:U-add2}
\end{align}
With the Schr\"odinger equation and Eq.~(\ref{eq:Lambda-F}), Eq.~(\ref{eq:U-add1})
becomes 
\begin{equation}
\text{Tr}\left[H(T)F(T)\right]=1.
\end{equation}
Eq.~(\ref{eq:U-add2}) is satisfied automatically since $\text{Tr}[F(t)]=0$
at all times.

\section{\label{sec:Derivation-Governing-Eqs}Derivation of the governing
equations for the QB equation }

In the most general case, the constraints are $\text{Tr}[H(t)]=0$
and Eqs.~(\ref{eq:norm-constraint})-(\ref{eq:term-constraint})
in the main text. Therefore 
\begin{equation}
F(t)=\lambda_{0}(t)[H(t)+G(t)],\label{eq:F-restricted}
\end{equation}
where 
\begin{equation}
G(t)\equiv\frac{1}{\lambda_{0}(t)}\sum_{j\ge1}\lambda_{j}(t)\mathcal{X}_{j}\,.
\end{equation}
Now let us solve the CHKO equation Eq.~(\ref{eq:CHKO}). The solution
is
\begin{equation}
F(t)=U(t)F(0)U^{\dagger}(t),\label{eq:UFUdag}
\end{equation}
where
\begin{equation}
\text{i}\dot{U}(t)=H(t)U(t).
\end{equation}
Substituting Eq.~(\ref{eq:F-restricted}) into Eq.~(\ref{eq:UFUdag}),
we obtain 
\begin{equation}
\lambda_{0}(t)\left[H(t)+G(t)\right]=U(t)F(0)U^{\dagger}(t).
\end{equation}
Multiplying both sides by $U(t)$, we find 
\begin{equation}
\text{i}\dot{U}(t)+G(t)U(t)=\frac{1}{\lambda_{0}(t)}U(t)F(0).\label{eq:U-diff}
\end{equation}
Therefore, we consider the Hamiltonian. 
\begin{equation}
H(t)=\text{i}\dot{U}(t)U^{\dagger}(t)=-G(t)+\frac{1}{\lambda_{0}(t)}U(t)F(0)U^{\dagger}(t).\label{eq:H-preformal}
\end{equation}
Eq.~(\ref{eq:U-diff}) is a first-order differential equation
with initial condition $U(0)=\mathbb{I}$ and thus have unique
solution. When $M=1$, where $G(t)$ commutes at all times, the solution
to Eq.~(\ref{eq:U-diff}) is \textcolor{red}{}
\begin{equation}
U(t)=\exp\left[\text{i}\int_{0}^{t}G(\tau)d\tau\right]\exp\left[-\text{i}F(0)\int_{0}^{t}\frac{d\tau}{\lambda_{0}(\tau)}\right].
\end{equation}
For the case $M\ge2$, while $G(t)$ does not commute at all times, 
 the formal solution admits the form
\begin{equation}
U(t)=\mathcal{T}\exp\left[\text{i}\int_{0}^{t}G(\tau)d\tau\right]\exp\left[-\text{i}F(0)\int_{0}^{t}\frac{d\tau}{\lambda_{0}(\tau)}\right],\label{eq:U-formal}
\end{equation}
where $\mathcal{T}$ denotes the time-ordering operator. For the sake of simplicity,
we  denote 
\begin{equation}
V(t)\equiv\mathcal{T}\exp\left[\text{i}\int_{0}^{t}G(\tau)d\tau\right].\label{eq:V-def}
\end{equation}
and note that it satisfies Eq.~(\ref{eq:V-SE}) in the main text.
Therefore, upon substituting the expression for $F(0)$, we establish Eq.~(\ref{eq:U-expression})
in the main text. 

One can check that  Eq.~(\ref{eq:U-formal}) is indeed a solution to Eq.~(\ref{eq:U-diff}),
given that
\begin{align}
\dot{U}(t) & =\text{i}G(t)U(t)-\text{i}\frac{1}{\lambda_{0}(t)}U(t)F(0).\label{eq:U-dot}
\end{align}
Thus, Eq.~(\ref{eq:H-preformal}) becomes 
\begin{equation}
H(t)=\frac{1}{\lambda_{0}(t)}V(t)F(0)V^{\dagger}(t)-G(t).\label{eq:H-formal}
\end{equation}
Upon substituting the expression for $F(0)$ and $G(t)$, we find
Eq.~(\ref{eq:H-expression}) in the main text. 

Eq.~(\ref{eq:H-formal}) satisfies the constraints $\text{Tr}[H(t)]=0$
and by default given Eq.~(\ref{eq:F-restricted}). \textcolor{red}{}
To satisfy the norm constraint~(\ref{eq:norm-constraint}) in the
main text, we compute 
\begin{equation}
\text{Tr}[H^{2}(t)]=\text{Tr}[G^{2}(t)]+\frac{\text{Tr}[F^{2}(0)]}{\lambda_{0}^{2}(t)}-\frac{2}{\lambda_{0}(t)}\text{Tr}[G(t)F(t)]=2\omega^{2},\label{eq:TrH2-explicit}
\end{equation}
where we have used the fact that
\begin{equation}
F(t)=U(t)F(0)U^{\dagger}(t)=V(t)F(0)V^{\dagger}(t).
\end{equation}
We note that 
\begin{align}
\text{Tr}[G^{2}(t)] & =N\frac{\sum_{j\ge1}\lambda_{j}^{2}(t)}{\lambda_{0}^{2}(t)},\\
\text{Tr}[F^{2}(t)] & =2\omega^{2}\lambda_{0}^{2}(t)+N\sum_{j\ge1}\lambda_{j}^{2}(t),
\end{align}
and that $\text{Tr}[F^{2}(t)]$ is a conserved quantity, so we have 
\begin{equation}
2\omega^{2}\lambda_{0}^{2}(t)+N\sum_{j\ge1}\lambda_{j}^{2}(t)=2\omega^{2}\lambda_{0}^{2}(0)+N\sum_{j\ge1}\lambda_{j}^{2}(0).\label{eq:TrF2}
\end{equation}
Thus, Eq.~(\ref{eq:TrH2-explicit}) becomes
\begin{equation}
\text{Tr}[G(t)F(t)]=\frac{N\sum_{j\ge}\lambda_{j}^{2}(t)}{\lambda_{0}(t)}.\label{eq:TrYFt}
\end{equation}
To satisfy the constraint~(\ref{eq:term-constraint}), we have 
\begin{equation}
\text{Tr}[H(t)\mathcal{X}_{j}]=\frac{1}{\lambda_{0}(t)}\left(\text{Tr}\left[\mathcal{X}_{j}F(t)\right]-\lambda_{j}(t)N\right)=0,
\end{equation}
which leads to 
\begin{equation}
\text{Tr}\left[\mathcal{X}_{j}F(t)\right]=N\lambda_{j}(t),\,\forall j\ge1.\label{eq:TrXFt}
\end{equation}
Since Eq.~(\ref{eq:TrXFt}) implies Eq.~(\ref{eq:TrYFt}), we note
that only Eq.~(\ref{eq:TrF2}) and Eq.~(\ref{eq:TrXFt}) are independent.
Since Eq.~(\ref{eq:lam-0}) holds at the initial time, it will be
satisfied if its first derivatives on both sides are equal at all
times. This leads to the following differential equation 
\begin{equation}
2\omega^{2}\lambda_{0}(t)\dot{\lambda}_{0}(t)+N\sum_{j\ge1}\lambda_{j}(t)\dot{\lambda}_{j}(t)=0.\label{eq:lam0-dot}
\end{equation}
Similarly, Eq.~(\ref{eq:TrXFt}) is satisfied initially. So taking
time derivatives on both sides of Eq.~(\ref{eq:TrXFt}) and using
the fact that 
\[
\dot{F}(t)=-\text{i}[H(t),\,F(t)]=-\text{i}\lambda_{0}(t)[H(t),\,G(t)],
\]
 together with the identity $\text{Tr}\left(A[B,\,C]\right)=\text{Tr}\left(C[A,\,B]\right)$, 
 yields the differential equation 
\begin{equation}
\dot{\lambda}_{j}(t)=\frac{1}{N}\sum_{l\ge1}\lambda_{l}(t)\eta_{jl}(t),\label{eq:lamj-dot}
\end{equation}
where 
\begin{equation}
\eta_{jl}(t)=\text{Tr}[H(t)\mathscr{X}_{jl}].\label{eq:eta-def}
\end{equation}
 Substituting Eq.~(\ref{eq:lamj-dot}) into Eq.~(\ref{eq:lam0-dot}),
we obtain Eq.~(\ref{eq:lam-0}) in the main text. 

\section{\label{sec:Initial-form}Determine the form of $F(0)$ and $H(0)$
using Eq.~(\ref{eq:CHKO-initial}) for the two-qubit example}

According to Eq.~(\ref{eq:H0-choice}), one should choose the initial
Hamiltonian as,
\begin{align}
H(0) & =\mu_{22}\sigma_{1}^{2}\sigma_{2}^{22}+\mu_{33}\sigma_{1}^{3}\sigma_{2}^{3}\nonumber \\
 & +\mu_{23}\sigma_{1}^{2}\sigma_{2}^{3}+\mu_{32}\sigma_{1}^{3}\sigma_{2}^{2},\label{eq:H0-2qu}
\end{align}
and 
\begin{align}
F(0) & =H(0)+\lambda_{10}\sigma_{1}^{1}+\lambda_{20}\sigma_{1}^{2}+\lambda_{30}\sigma_{1}^{3}\nonumber \\
 & +\lambda_{01}\sigma_{2}^{1}+\lambda_{02}\sigma_{2}^{2}+\lambda_{03}\sigma_{2}^{3}+\lambda_{11}\sigma_{1}^{1}\sigma_{2}^{1}\nonumber \\
 & +\lambda_{12}\sigma_{1}^{1}\sigma_{2}^{2}+\lambda_{13}\sigma_{1}^{1}\sigma_{2}^{3}+\lambda_{21}\sigma_{1}^{2}\sigma_{2}^{1}+\lambda_{31}\sigma_{1}^{3}\sigma_{2}^{1},\label{eq:F0-2qu}
\end{align}
where we have set $\lambda_{0}=1$. We consider $\ket{e_{1}}=\ket{\psi_{i}}=\ket{11}$
and $\ket{e_{2}}=\ket{\psi_{f}^{\perp}}=\ket{00}$. Note that $\braket{e_{k}\big|\sigma_{1}^{\alpha}\sigma_{2}^{\beta}\big|e_{k}}=0$
as long as $\alpha,\,\beta$ are both not equal to $3$. 
Therefore, we find that $\braket{e_{k}\big|\mathcal{X}_{j}\big|e_{k}}=0$
except for $\mathcal{X}_{j}=\sigma_{1}^{3}$ or $\sigma_{2}^{3}$. Using
$\braket{e_{1}\big|F(0)\big|e_{1}}=0$ leads to 
\begin{equation}
\mu_{33}-\lambda_{30}-\lambda_{03}=0.
\end{equation}
Similarly, for $k\ge2$,  $\braket{e_{k}\big|F(0)\big|e_{k}}=0$ lead
to
\begin{align}
\mu_{33}+\lambda_{30}+\lambda_{03} & =0,\\
\mu_{33}-\lambda_{30}+\lambda_{03} & =0,\\
\mu_{33}+\lambda_{30}-\lambda_{03} & =0,
\end{align}
from which we find Eq.~(\ref{eq:2qu-coeff1}) in the main text. 

One can also readily find that
\begin{equation}
\begin{array}{ccc}
\braket{00\big|\sigma_{1}^{\alpha}\sigma_{2}^{\beta}\big|10}=0,\,\beta\neq3,\, & \braket{00\big|\sigma_{1}^{1}\sigma_{2}^{3}\big|10}=1,\, & \braket{00\big|\sigma_{1}^{2}\sigma_{2}^{3}\big|10}=-\text{i},\end{array}\label{eq:id1}
\end{equation}
which leads to $\braket{e_{2}\big|F(0)\big|e_{3}}=-\text{i}(\mu_{23}+\lambda_{20})+\lambda_{13}+\lambda_{10}=0$.
On the other hand, we note $\mu_{\alpha\beta}$ and $\lambda_{\gamma\delta}$
must be real, so we obtain Eq.~(\ref{eq:2qu-coeff2}) in the main
text. 

Similarly, we observe 
\begin{equation}
\begin{array}{ccc}
\braket{00\big|\sigma_{1}^{\alpha}\sigma_{2}^{\beta}\big|01}=0,\,\beta\neq3,\, & \braket{00\big|\sigma_{1}^{3}\sigma_{2}^{1}\big|01}=1,\, & \braket{00\big|\sigma_{1}^{3}\sigma_{2}^{2}\big|01}=\text{i},\end{array}\label{eq:id2}
\end{equation}
which leads to Eq.~(\ref{eq:2qu-coeff4}) in the main text. At this
point, we find 
\begin{equation}
H(0)=\mu_{22}\sigma_{1}^{2}\sigma_{2}^{2}+\mu_{23}\sigma_{1}^{2}\sigma_{2}^{3}+\mu_{32}\sigma_{1}^{3}\sigma_{2}^{2}
\end{equation}
and 
\begin{align}
F(0) & =\mu_{22}\sigma_{1}^{2}\sigma_{2}^{2}++\mu_{23}\sigma_{1}^{2}\sigma_{2}^{3}+\mu_{32}\sigma_{1}^{3}\sigma_{2}^{2}\nonumber \\
 & +\lambda_{10}\sigma_{1}^{1}+\lambda_{20}\sigma_{1}^{2}+\lambda_{01}\sigma_{2}^{1}+\lambda_{02}\sigma_{2}^{2}\nonumber \\
 & +\lambda_{10}(\sigma_{1}^{1}-\sigma_{1}^{1}\sigma_{2}^{3})+\lambda_{01}(\sigma_{2}^{1}-\sigma_{1}^{3}\sigma_{2}^{1})\nonumber \\
 & +\lambda_{11}\sigma_{1}^{1}\sigma_{2}^{1}+\lambda_{12}\sigma_{1}^{1}\sigma_{2}^{2}+\lambda_{21}\sigma_{1}^{2}\sigma_{2}^{1}.
\end{align}
Finally, we note that
\begin{equation}
\begin{array}{cccc}
\braket{10\big|\sigma_{1}^{2}\sigma_{2}^{2}\big|01}=1,\, & \braket{10\big|\sigma_{1}^{1}\sigma_{2}^{2}\big|01}=1,\, & \braket{10\big|\sigma_{1}^{1}\sigma_{2}^{2}\big|01}=-\text{i},\, & \braket{10\big|\sigma_{1}^{2}\sigma_{2}^{1}\big|01}=\text{i},\end{array}\label{eq:id3}
\end{equation}
which leads to Eq.~(\ref{eq:2qu-coeff6}) in the main text. 

\end{widetext}

\bibliographystyle{apsrev4-1}
\bibliography{QBE}

\end{document}